\documentclass[11pt,a4paper]{article}

\usepackage[margin=1in]{geometry}
\usepackage[
  bookmarks=true,
  bookmarksnumbered=true,
  bookmarksopen=true,
  pdfborder={0 0 0},
  breaklinks=true,
  colorlinks=true,
  linkcolor=black,
  citecolor=black,
  filecolor=black,
  urlcolor=black,
]{hyperref}
\usepackage{amssymb,amsmath,amsthm}
\usepackage{palatino}
\usepackage{ctable}
\usepackage[capitalise]{cleveref}
\usepackage{enumitem}
\usepackage{algorithm}
\usepackage{algpseudocode}
\usepackage[hypcap=false]{caption}
\usepackage[round]{natbib}

\setlist{nolistsep}

\DeclareCaptionFormat{algor}{%
  \hrulefill\par\offinterlineskip\vskip1pt%
    \textbf{#1#2}#3\offinterlineskip\hrulefill}
\DeclareCaptionStyle{algori}{singlelinecheck=off,format=algor,labelsep=space}
\algrenewcommand{\algorithmiccomment}[1]{\hfill {$\mathit{//}$} #1}
\algnewcommand{\LineComment}[1]{\State {\bf // #1}}

\usepackage{fnbreak} % Warn when a footnote spans multiple pages.

\newlist{types}{enumerate}{1}
\crefname{typesi}{Type}{Types}
\setlist[types,1]{label=\arabic*.,ref=\arabic*}

\newlist{parts}{enumerate}{2}
\crefname{partsi}{Part}{Parts}
\setlist[parts,1]{label=\arabic*.,ref=\arabic*}
\setlist[parts,2]{label=\roman*.,ref=\roman*}

\newlist{properties}{enumerate}{1}
\crefname{propertiesi}{Property}{Properties}
\setlist[properties,1]{label=\arabic*.,ref=\arabic*}
\setlist[properties,2]{label=\roman*.,ref=\roman*}

\newlist{modifications}{enumerate}{2}
\crefname{modificationsi}{Modification}{Modifications}
\setlist[modifications,1]{label=\arabic*.,ref=\arabic*}
\setlist[modifications,2]{label=\roman*.,ref=\roman*}

\theoremstyle{plain}
\newtheorem{thm}{Theorem}
\crefname{thm}{Theorem}{Theorems}
\newtheorem{prop}[thm]{Proposition}
\newtheorem{cor}{Corollary}
\crefname{cor}{Corollary}{Corollaries}
\newtheorem{lemma}{Lemma}
\newtheorem{claim}[lemma]{Lemma}

\theoremstyle{definition}
\newtheorem{defn}{Definition}
\newtheorem{ex}{Example}
\crefname{ex}{Example}{Examples}
\newtheorem{remark}{Remark}

\newcommand{\crefpart}[2]{\cref{#1}(\labelcref{#1-#2})}

\newcommand{\crefshowpart}[2]{\cref{#1-#2} of \cref{#1}}

\newcommand{\refintitle}[1]{\texorpdfstring{\ref{#1}}{\ref*{#1}}}

\newcommand{\eqdef}{\triangleq}
\newcommand{\naturals}{\mathbb{N}}

\newcommand{\matching}{\mu}

\newcommand{\comp}[1]{#1^{\mathsf{c}}}

\newcommand{\tobematched}[1]{\tilde{#1}}

\newcommand{\prefs}[1]{\mathcal{P}_{#1}}
\newcommand{\tildeprefs}[1]{\tilde{\mathcal{P}}_{#1}}

\newcommand{\matched}[1]{#1_{\matching}}
\newcommand{\unmatched}[1]{\comp{#1_{\matching}}}

\newcommand{\resultmatching}{\matching_{\prefs{W},\prefs{M}}'}
\newcommand{\tilderesultmatching}{\matching_{\tildeprefs{W},\prefs{M}}'}
\newcommand{\wresultmatching}{\matching_{\tildeprefs{W}^w,\prefs{M}}'}
\newcommand{\iincresultmatching}{\matching_{\prefs{W}^{i+1},\prefs{M}}^i}
\newcommand{\ipresultmatching}{\matching_{\prefs{W}',\prefs{M}}^i}

\newcommand{\menopt}{\mathit{MenOpt}}

\newcommand{\generalcycle}{(w_1 \xrightarrow{m_1} w_2 \xrightarrow{m_2} \cdots w_{d-1} \xrightarrow{m_{d-1}} w_d)}
\newcommand{\generalcycleprime}{(w_1' \xrightarrow{m_1'} w_2' \xrightarrow{m_2'} \cdots w_{d'-1}' \xrightarrow{m_{d'-1}'} w_{d'}')}

\newcommand{\abstractcycle}[1]{C_{\matching'}^{\matching}(#1)}

\newcommand{\rejectcycle}{C_{\matching'}^{\prefs{W},\prefs{M}}}
\newcommand{\tilderejectcycle}{C_{\matching'}^{\tildeprefs{W},\prefs{M}}}

\newcommand{\twrejectcycle}{C_{\matching'}^{\prefs{W}^{\tilde{w}},\prefs{M}}}
\newcommand{\hwrejectcycle}{C_{\matching'}^{\prefs{W}^{\hat{w}},\prefs{M}}}

\newcommand{\fullrun}{R^{\prefs{W},\prefs{M}}}

\newcommand{\zfullrun}{R^{\prefs{W}^0,\prefs{M}}}
\newcommand{\ofullrun}{R^{\prefs{W}^1,\prefs{M}}}
\newcommand{\ifullrun}{R^{\prefs{W}^i,\prefs{M}}}
\newcommand{\incfullrun}{R^{\prefs{W}^{i+1},\prefs{M}}}
\newcommand{\pfullrun}{R^{\prefs{W}',\prefs{M}}}

\newcommand{\run}{R_{\matching'}^{\prefs{W},\prefs{M}}}
\newcommand{\tilderun}{R_{\matching'}^{\tildeprefs{W},\prefs{M}}}

\newcommand{\twrun}{R_{\matching'}^{\prefs{W}^{\tilde{w}},\prefs{M}}}
\newcommand{\hwrun}{R_{\matching'}^{\prefs{W}^{\hat{w}},\prefs{M}}}
\newcommand{\iincrun}{R_{\matching^i}^{\prefs{W}^{i+1},\prefs{M}}}
\newcommand{\iprun}{R_{\matching^i}^{\prefs{W}',\prefs{M}}}
\newcommand{\pprun}{R_{\matching'}^{\prefs{W}',\prefs{M}}}

\newcommand{\firsti}{\tilde{R}_{1/2}^i}
\newcommand{\secondi}{\tilde{R}_{2/2}^i}

\newcommand{\firstpart}{\fullrun_{1/2}}
\newcommand{\secondpart}{\fullrun_{2/2}}

\newcommand{\firstip}{R_{1/2}^{i+1}}
\newcommand{\secondip}{R_{2/2}^{i+1}}

\newcommand{\firstthird}{R_{1/3}^{i+1}}
\newcommand{\secondthird}{R_{2/3}^{i+1}}
\newcommand{\thirdthird}{R_{3/3}^{i+1}}

\hypersetup{
  pdfauthor      = {Yannai A. Gonczarowski <yannai@gonch.name>},
  pdftitle       = {Manipulation of Stable Matchings using Minimal Blacklists},
}

\title{Manipulation of Stable Matchings \\ using Minimal Blacklists}
\author{Yannai A. Gonczarowski\thanks{
Einstein Institute of Mathematics,
Rachel and Selim Benin School of Computer Science and Engineering
and Center for the Study of Rationality,
Hebrew University of Jerusalem, Israel; and Microsoft Research. \mbox{\emph{Email}: \href{mailto:yannai@gonch.name}{yannai@gonch.name}}.
}}
\date{March 5, 2014}

\begin{document}
\begin{titlepage}

\maketitle

\begin{abstract}
\cite{Gale-Sotomayor-ms-machiavelli} have shown that in the
Gale-Shapley matching algorithm \citeyearpar{Gale-Shapley}, the proposed-to side $W$ (referred to as \emph{women} there)
can strategically force the \mbox{$W$-optimal} stable
matching as the $M$-optimal one by truncating their preference lists, each woman
possibly blacklisting \emph{all but one man}.
As \citeauthor{GI89} have already noted in \citeyear{GI89},
no results are known regarding achieving this feat by means other than such preference-list truncation, i.e.\ by
also permuting preference lists.

We answer \citeauthor{GI89}'s open question by providing
tight upper bounds on the amount of blacklists and their combined size,
that are required by the women to force a given matching
as the $M$-optimal stable matching, or, more generally, as the unique
stable matching. Our results show that the coalition of all women can
strategically force any matching as the unique stable matching, using preference
lists in which at most half of the women have nonempty blacklists, and in
which the average blacklist size is \emph{less than~$1$}. \linebreak This allows the women to
manipulate the market in a manner that is far more inconspicuous, in a sense,
than previously realized.
When there are less women than men, we show that in the absence of blacklists
for men,
the women can force any matching as the unique stable matching \emph{without
blacklisting anyone}, while when there are more women than men, each
to-be-unmatched woman may have to blacklist as many as all men.
Together, these results shed light on the question of how much, if at all,
do given preferences for
one side \emph{a priori} impose limitations on the set of stable matchings
under various conditions.
All of the results in this paper are constructive, providing efficient algorithms
for calculating the desired strategies.
\end{abstract}

\vfill
\noindent
\textbf{Keywords}:
Matching; Stability; Deferred Acceptance; Manipulation; Blacklist
\vfill
\vfill
\pagenumbering{Roman}
\thispagestyle{empty}
\end{titlepage}

\pagenumbering{arabic}

\section{Introduction}

\cite{Gale-Shapley}, in their seminal paper, have introduced the question of 
finding a matching between a set of women $W$ and a set of men $M$, that is,  a one-to-one mapping between a subset of $W$ (the women \emph{matched} by this matching) and a subset of $M$ (the men matched by this matching), that satisfies certain desirable properties.\footnote{\citeauthor{Gale-Shapley} have studied one-to-many matchings that generalize the scenario
described here, as well as one-to-one matchings under the condition $|W|\!=\!|M|$. The scenario that we describe, of one-to-one matchings between sets of possibly-unequal size, was first explicitly studied by~\cite{MW70}.} Such $W$ and $M$, along with preferences for each participant, comprise a two-sided \emph{matching market}. Over the past few decades, applications of
matching markets have become widespread, from assigning medical interns to hospitals~\citep{Roth-labor-evolution} and students to high schools~\citep{nyc-match,boston-match}, to centralizing kidney donation assignments~\citep{kidneys}.

In \citeauthor{Gale-Shapley}'s model, each woman has a strict order of preference over a subset of~$M$.
This subset, ordered according to her order of preference, is called the \emph{preference list} of this woman,
and is interpreted as those men with whom she would prefer to be matched over being matched with no-one. Each man similarly has a preference list of women.

\begin{defn}[Blacklist]
The \emph{blacklist} of a woman $w$ is the set of all men that are not in her preference list, i.e.\ the men that $w$ finds unacceptable,
even when the alternative is being unmatched (i.e.\ being matched with no-one). The blacklist of a man is defined analogously.
\end{defn}

\begin{defn}[Individual Rationality]
A matching is \emph{$M$-rational} if no man is matched with a woman from his blacklist;
it is \emph{$W$-rational} if no woman is matched with a man from her blacklist.
\end{defn}

A matching is \emph{unstable} under the given orders of preference if either~~\emph{(i)} it is not
$W$-rational, or not $M$-rational, or~~\emph{(ii)} there exist a woman and a man, each of whom prefer being matched with the other over the partner (or lack thereof) assigned to them by the matching. A matching that is not unstable is \emph{stable}; \cite{Gale-Shapley} have proved that a stable matching always exists, and have presented
an efficient algorithm for finding such a matching.

While in some cases only a single stable matching exists, in general there are many stable matchings, differing significantly in the outcome for the various participants~\citep{PittelAverageNumber}.
\citeauthor{Gale-Shapley} have shown that their algorithm finds the $M$-optimal stable matching, i.e.\ a stable matching over which no man prefers any other stable matching. \cite{MW71} have additionally shown that this matching is at the same time the $W$-worst stable matching, i.e.\ a stable matching that is worst for each woman.
The benefits of the Gale-Shapley matching algorithm for the men have been
demonstrated even further by \cite{Dubins-Freedman},
who have shown that no man can unilaterally strategically manipulate this algorithm (i.e.\ manipulate the $M$-optimal stable matching) to his advantage and, moreover, that no coalition of men
can all benefit  from jointly manipulating it.\footnote{The interested reader is referred
to~\cite{DGS87} for an even
stronger nonmanipulability theorem, and
to~\cite{H06} for the study
of lying by men in a probabilistic setting.}
These observations have led to the study by \cite{Gale-Sotomayor-ms-machiavelli} of strategic manipulations of the algorithm, and of stable matchings in general, by women.

\cite{Gale-Sotomayor-ms-machiavelli} have shown that if more than one stable matching exists, then at least one woman
can unilaterally manipulate the $M$-optimal stable matching to her own advantage.
Moreover, they have shown that the coalition of all women can force the $W$-optimal stable matching as the unique stable matching, and thus as the $M$-optimal stable matching (and the outcome of the Gale-Shapley matching
algorithm), by truncating their preference lists, each woman blacklisting all men over whom she prefers her $W$-optimal stable partner.
Note that this strategy may lead to each woman blacklisting all but one man (i.e.\ declaring a blacklist of size \mbox{$|M|\!-\!1$}), which makes the women's manipulation painfully obvious to 
any observer of their submitted preferences.
Despite the extensive literature regarding stable matchings and manipulation, as \citet[pp.\ 59,65]{GI89} have already noted in \citeyear{GI89},
no results are known regarding achieving this feat by means other than such preference-list truncation, i.e.\ by
also permuting preference lists.

As it is widely believed that successful matching mechanisms should produce stable outcomes~\citep{roth-economist-engineer},
many real-life matching mechanisms (even those not based on \citeauthor{Gale-Shapley}'s algorithm)
are designed with the goal of producing stable matchings.
Therefore, the ability to force some matching as the unique stable matching is of broad significance.

In \cref{blacklists}, we answer \citeauthor{GI89}'s open question by tightly characterizing the worst-case amount of blacklists and their sizes,
that are required by the women to force the $W$-optimal stable matching (or more generally, any $M$-rational matching)
as the unique stable matching, and in particular as the outcome of the Gale-Shapley matching algorithm.
We start by analysing \emph{perfect} matchings, i.e.\ matchings in which no participant is unmatched.
A corollary of our results for this case is as follows.

\begin{thm}[Weaker version of \cref{inconspicuous}; the latter is shown to be tight in \cref{inconspicuous-tight}]
\label{intro-thm}
~
\begin{enumerate}
\item
When $|W|\!=\!|M|$, the coalition of all women can force any $M$-rational perfect matching
as the \emph{unique} stable matching,
using a profile of preference lists in which at most half of the women have blacklists, and in which the average blacklist size
is less than 1.
\item This profile of preference lists can be computed in
$O(n^3)$ time.
\end{enumerate}
\end{thm}

This low upper bound of less than $1$ on the average blacklist size that is guaranteed by \cref{intro-thm} should be contrasted with the previously-known upper bound of $|M|\!-\!1$,
which is attainable using truncation,
as mentioned above. Indeed, our results provide a far more ``inconspicuous'' manipulation than the
one suggested by \citeauthor{Gale-Sotomayor-ms-machiavelli} using preference truncating, even though both manipulations result in the same matching. Thus, if women have to pay even a small cost
for every man that they blacklist, then our solution beats preference truncation by an order of magnitude.
The manipulation we suggest becomes even more inconspicuous in implementations where, by design, participants
are required to submit preference lists of at most a certain length, as in the case in some real-life matching markets~\citep[see e.g.][]{nyc-match}.

It turns out that the case of a \emph{balanced market}, i.e.\ $|W|\!=\!|M|$, is a singular case at which a phase change occurs. 
When there are less women than men, we show (see \cref{inconspicuous-not-all-matched,tightness-not-all-matched})
that in the absence of blacklists for men,
the women can always force any $M$-rational matching, in which all women are matched, as the unique
stable matching \emph{without blacklisting anyone}. In contrast, when there are more women than men
(or more generally, when not all women should be matched), each to-be-unmatched woman may have to blacklist as many as all men.
Together, all of these results shed light on the question of how much, if at all, do given preferences for one side \emph{a priori} impose limitations on the set
of stable matchings under various conditions.

It is interesting to note
that \cite{Ashlagi-Leshno} have shown that a somewhat
similar phase change occurs w.r.t.\ the expected ranking of each participant's partner in the participant's
preference list in a random market. Both that paper and this one show, in different senses,
that the preferences of the smaller side of the market, even if it is only
slightly smaller, play a far more
significant role than may be expected in determining the stable matchings, and those of the larger side ---
a considerably insignificant one. In a sense, our results extend this qualitative statement from a random market to \emph{any} market.

Another implication of our results is
in goods allocation problems, i.e.\ matching problems in which only one side (the \emph{buyers}) has
preferences, while the other side (the \emph{goods}) does not.
\cite{atila-school-choice} and \cite{atila-strategyproofness-efficiency} consider using a version of the (student-optimal)
Gale-Shapley algorithm in order to allocate school seats to children;
as school priorities are very coarse \citep[and sometimes nonexistent, see e.g.][]{nyc-match}, some tie-breaking rule is needed so that the algorithm can be run.
Both of these papers advocate the use of a single tie-breaking rule for all schools (e.g.\ a single lottery conducted once before the algorithm runs, to assign each student
a number to be used for tie breaking at every school) over a different tie-breaking rule for each school (e.g.\ a different lottery for each school), arguing
that the former results in higher social welfare. Our results strengthen this argument and make it even more concrete, showing that in goods allocation problems,
an implementation using multiple arbitrary (e.g.\ randomized) tie-breaking rules may yield
any allocation that is buyer-rational,
regardless of the internal orders of the buyer's preference lists, if even an amortized less-than-one-buyer blacklist per good is allowed (or even in the
absence any blacklists for the goods whatsoever, if there are more buyers than goods); in contrast, it is easy to see that an implementation with a single tie-breaking
rule is equivalent to serial dictatorship~\citep{rsd}, and thus e.g.\ constructs Pareto-efficient allocations.

While \cref{intro-thm} guarantees an average blacklist size of less than~$1$, it provides no guarantee regarding the size of each individual blacklist.
Indeed, it may be the case (see \cref{inconspicuous-tight}) that in order to force a particular matching, one of the women is required to blacklist all but one man, while all the other women's blacklists are empty.
In \cref{divorces}, we introduce a natural variant of the Gale-Shapley algorithm, in which blacklisting is a dynamic process called \emph{divorcing};
this variant, which we call the \emph{Gale-Shapley algorithm with divorces}, is exactly as manipulable as the unaltered Gale-Shapley algorithm
(see \cref{divorce-equiv-blacklist} for the precise sense).
We characterize the worst-case pattern of divorces
in this variant, and show that
the coalition of all women can force any matching as the outcome by at most $n\!-\!1$ women
getting divorced, each exactly once
(see \cref{one-divorce-per-woman}).
This result is shown to be tight as well (see \cref{divorce-tight}).

All of the results in this paper are constructive, providing efficient algorithms for calculating the desired strategies; full proofs for all theorems are given in the appendix.
While the proofs are involved, the resulting algorithms are relatively simple to describe.

\section{Additional Related Work}

A participant with a nonempty blacklist is often said in the literature to have an \emph{incomplete preference list} (not to be confused with a partial order of preferences, which indicates an indifference of sorts between certain alternatives).
\cite{KobayashiMatsui} study the same problem as we do, yet confined only to a balanced market, and give a necessary and sufficient condition (on the men's preferences and the desired matching) for a solution with no blacklists whatsoever to exist; they also give a polynomial-time algorithm for computing a solution (i.e.\ women's strategies) when this condition holds.

Manipulations by a single woman of the (men-proposing) Gale-Shapley algorithm have been extensively studied. Notably relevant to our setting is the study of \cite{TST01}, who give an efficient algorithm for finding an optimal strategy (preference list) for a single woman (with given true preferences) when blacklists are not allowed, given the submitted preference lists of the rest of the market; they note that it is not always possible for such a woman to obtain her top choice under this setting. They also give simulation results suggesting that the proportion of women who can gain by manipulation in a random market approaches zero as the market grows; \cite{KojimaPathak} indeed prove that even when blacklists are allowed and even in a many-men-to-one-woman setting, the proportion of women who can gain by manipulation in a random market, when the lengths of men's preference lists are fixed, tends to zero as the market grows.
In a contrast of sorts, since \cite{PittelLikelySolutions} shows that in a random market each participant w.h.p.\ has more than
one stable partner, each woman w.h.p.\ benefits (and no woman is ever harmed) if
all women form a coalition and force the $W$-optimal stable matching.

\sloppypar{
Bridging the gap between manipulation by a single woman and manipulation by all women is manipulation by a coalition of women.
\cite{Sisterhood} show that even when blacklists are allowed and even in many-to-many settings, if a coalition of women manipulates the (men-proposing) Gale-Shapley algorithm in a way that harms none of them, then no truthful woman is harmed either and no man gains.
Coalitional manipulation is also studied by \cite{H06}, who studies incentives for manipulation of the (men-proposing) Gale-Shapley algorithm by coalitions of men.}

\cite{PRVWAAMAS2011} show that in the absence of blacklists, there exists a stable mechanism that is computationally-hard to manipulate by a single participant. As our results yield a unique stable matching, they hold for any stable mechanism, and so show in a sense that their results do not carry over to the case of manipulation by an entire side of the market
if even very small blacklists are allowed. \cite{Eeckhout} gives a sufficient condition for uniqueness of a stable matching in the absence of blacklists; our results produce a unique stable matching, yet do not meet this sufficient condition, even if this condition is extended to a setting with blacklists by adding additional participants denoting an unmatched status.

\section{Model and Preliminaries}\label{preliminaries}

In order to standardize the notation used throughout this paper, and
for self-containment, let us quickly recapitulate the definitions from the previous section, as well as the
Gale-Shapley algorithm~\citeyearpar{Gale-Shapley}. Let $W$ and $M$ be disjoint finite sets, referred to as the sets of women and men, respectively.

\begin{defn}[Preference Lists; Blacklists; Combined Blacklist Size]\leavevmode
\begin{enumerate}
\item
A \emph{preference list} for a woman $w \in W$ is a totally-ordered subset of $M$.
\item
A \emph{profile of preference lists} for $W$ is a specification of a preference list for each $w\!\in\!W$.
\item
Given a profile $\prefs{W}$ of preference lists for $W$, the \emph{blacklist} of a woman $w \in W$, denoted by $B_w(\prefs{W})$, is the set of men in $M$ who do not appear in
$w$'s preference list (i.e.\ all men who are declared unacceptable by $w$). We say that $w$ \emph{blacklists} the men in $B_w(\prefs{W})$.
\item
The \emph{combined size} of all blacklists in $\prefs{W}$ is defined to be
$\sum_{w \in W}\bigl|B_w(\prefs{W})\bigr|$.
\end{enumerate}
We define preference lists, profiles of preference list, and blacklists for $M$ analogously.
\end{defn}

\begin{defn}[Matching]\leavevmode
A \emph{matching} is a one-to-one mapping between a subset of $W$ and a subset of $M$. By slight abuse of notation, we denote the woman matched with a man $m$ by a matching $\matching$ by
$\matching(m)$ instead of $\matching^{-1}(m)$.
Given a matching $\matching$, we define the subset of $W$ (resp.~$M$) on which $\matching$ is defined as $\matched{W}$ (resp.~$\matched{M}$).
If $\matched{W}=W$ and $\matched{M}=M$, then we say that~$\matching$ is \emph{perfect}; otherwise, we say that~$\matching$ is \emph{partial}.
\end{defn}

Throughout the remainder of this section, let $\prefs{W}$ and $\prefs{M}$ be profiles of preference lists for~$W$ and $M$, respectively.

\begin{defn}[Individual Rationality; Stability]
Let $\matching$ be a matching.
\begin{enumerate}
\item
$\matching$ is said to be \emph{$G$-rational}, for a side $G\!\in\!\{W, M\}$, if  $\matching(p)\!\notin\!B_p(\prefs{G})$ for every matched participant $p\!\in\!\matched{G}$.
\item
$\matching$ is said to be \emph{unstable} w.r.t\ $\prefs{W}$ and $\prefs{M}$, if either it is not $W$-rational or not $M$-rational, or there exist a woman $w\!\in\!W$ and a man $m\!\in\!M$ that do not blacklist each other, each of whom either unmatched in $\matching$, or preferring the other over the partner matched to them by $\matching$.
\item
If $\matching$ is not unstable, then it is said to be \emph{stable}.
\end{enumerate}
\end{defn}

\begin{defn}[The Gale-Shapley Algorithm~\citeyearpar{Gale-Shapley}]
The following algorithm is henceforth referred to as the \emph{Gale-Shapley
algorithm}:\footnote{The semantics of the algorithm that we describe w.r.t.\ blacklists are slightly different than those of \citeauthor{Gale-Shapley}'s original algorithm.
Nonetheless, the outcome is the same by a theorem of \cite{Dubins-Freedman}.}
The algorithm is divided into steps, to which we
refer as \emph{nights}. On each night, each man serenades under the window of the
woman he prefers most among all women who have not (yet) rejected him and who are not in his blacklist (if such a woman exists), and
then each woman, under whose window more than one man serenades, rejects every
man who serenades under her window, except for the man she prefers most
among these men; any woman under whose window serenade only men from her blacklist, rejects all of these men.
The algorithm stops on a night on which no man is rejected by
any woman, and then each woman under whose window a man has serenaded on this night is matched to this man.
Other women are unmatched (as are men who have not serenaded under any window on this night).
\end{defn}

\begin{thm}[\cite{Gale-Shapley}]\label{men-optimal}
The Gale-Shapley algorithm stops and yields a stable matching between $W$
and $M$ (in particular, such a matching exists), and this stable matching is $M$-optimal, i.e.\ no stable matching
is better for any man.
\end{thm}

Theorem~\ref{men-optimal} states that the stable matching given by the
Gale-Shapley algorithm
is optimal (out of all stable matchings) for each man. A conceptually-reverse
claim holds regarding the women:

\begin{thm}[\cite{MW71}]\label{women-worst}
No stable matching is worse for any woman
than the matching given by the Gale-Shapley algorithm.
\end{thm}

In particular, if both the Gale-Shapley algorithm as described above, and the Gale-Shapley algorithm with reverse roles (i.e.\ with women serenading to men,
yielding the $W$-optimal stable matching), yield the same matching, then this matching is the unique stable matching.

\section{Blacklists}\label{blacklists}

\subsection{Perfect Matchings}

\begin{ex}[Forcing by men of any matching as the $M$-optimal stable matching]\label{men-forcing}
Let $W$ and~$M$ be equal-sized sets of women and men, respectively.
Let $\matching$ be a matching and
let $\prefs{M}$ be any profile of preference lists for $M$ 
in which each man $m$'s top choice is $\matching(m)$. For any profile
$\prefs{W}$ of preference lists for $W$ according to which
$\matching$ is $W$-rational, the $M$-optimal stable matching
(yield by the Gale-Shapley algorithm) is~$\matching$.
\end{ex}

\cite{Gale-Sotomayor-ms-machiavelli} have shown that
women can strategically force the $W$-optimal stable matching as the $M$-optimal one by truncating
their preference lists, each woman blacklisting all men over whom she prefers her $W$-optimal stable partner (without
altering the order of the men remaining in her preference list).
If the top choices of all women are distinct, then the $W$-optimal stable matching
is for each woman to be matched with her top choice. \citeauthor{Gale-Sotomayor-ms-machiavelli}'s strategy in this case is for each woman to
blacklisting all men but her top choice. Such a profile of preference lists indeed allows the women to dictate
the $M$-optimal stable matching a strong sense, as demonstrated by the following example.

\begin{ex}[Na\"{\i}ve forcing by women of any matching as the $M$-optimal stable matching]\label{overkill}
Let~$W$ and $M$ be equal-sized sets of women and men, respectively.
Let $\matching$ be a matching and
let $\prefs{W}$ be the profile of preference lists for $W$ 
in which each woman blacklists all of $M$ except for $\matching(m)$.
For any profile
$\prefs{M}$ of preference lists for $M$ according to which
$\matching$ is $M$-rational,
the only stable matching (and thus the $M$-optimal stable matching)
is~$\matching$.
\end{ex}

As illustrated by \cref{men-forcing,overkill}, each side can
force any ``other-side-rational'' matching as the $M$-optimal stable matching,
however in order to do so, the women may need to submit far more specific, and far shorter, preference lists.
Indeed, a conspiracy as in \cref{overkill}, with~$n\!\eqdef\!|W|\!=\!|M|$ nonempty blacklists with a combined size of $n \cdot (n-1)$,
would be painfully obvious to anyone examining the women's submitted preference lists.
Despite the extensive literature regarding stable matchings and manipulation, as \citet[pp.\ 59,65]{GI89} have already noted in \citeyear{GI89},
no results are known regarding achieving this feat by means other than such preference-list truncation, i.e.\ by
also permuting preference lists.\footnote{%
Manipulations via preference-list truncations generally lend to easier analysis than general manipulations. As \cite{Gale-Sotomayor-ms-machiavelli} show,
if a woman's utility is determined solely by the identity of her partner, then preference-list truncation
is a (weakly) dominant strategy (as we show in this paper, this is no longer the case if it is preferable to blacklist as few men as possible).
The combination of these
has led to most of the relevant literature focusing on manipulation via preference-list truncation.
Nonetheless, some results do extend to arbitrary manipulation, most of them requiring nonstandard proof techniques. \cite[See e.g.][]{Sisterhood}.}

The main result of this paper provides an answer to \citeauthor{GI89}'s open question, and shows that if the preferences of all men are known
to the women, then the latter can
force any given $M$-rational matching (and in particular the $W$-optimal stable matching)
as the unique stable matching (and in particular as the $M$-optimal stable matching), using a far smaller
combined size for all blacklists, and with significantly less nonempty blacklists than by using preference-list truncation,
rendering the manipulation far less obvious in some sense. For implementations
where, by design, one of the sides is required to submit preference
lists of at most a certain length~\citep[see e.g.][]{nyc-match},
this result allows such a side to force
any matching that is ``other-side-rational'', in an even further-inconspicuous way.

\begin{thm}[Manipulation with Minimal Blacklists]\label{inconspicuous}
Let $W$ and $M$ be equal-sized sets of women and men, respectively. Define
$n\!\eqdef\!|W|\!=\!|M|$.
Let $\prefs{M}$ be a profile of preference lists for $M$.
For every $M$-rational perfect matching
$\matching$, there exists a profile $\prefs{W}$
of preference lists for $W$, s.t.\ all of
the following hold.
\begin{parts}
\item
The only stable matching, given $\prefs{W}$ and $\prefs{M}$, is $\matching$.
\item
The blacklists in $\prefs{W}$ are pairwise disjoint, i.e.\ no man appears in more than
one blacklist.
\item
$n_b$, the number of women who have nonempty blacklists in $\prefs{W}$,
is at most $\left\lfloor \frac{n}{2} \right\rfloor$.
\item
The combined size of all blacklists in $\prefs{W}$ is at most $n\!-\!n_b$,
i.e.\ at most the number of women who have empty blacklists.
\end{parts}
\begin{sloppypar}
Furthermore, $\prefs{W}$ can be computed in worst-case $O(n^3)$ time,
best-case $O(n^2)$ time
and average-case
(assuming $\matching$ is uniformly distributed given $\prefs{M}$)
$O(n^2 \log n)$ time.
\end{sloppypar}
\end{thm}

We note that \cref{inconspicuous} guarantees a combined blacklist size no greater than the size of each of the individual $n$ blacklists
from \cref{overkill} --- an order-of-magnitude improvement.

\begin{cor}[Upper Bound on Combined Blacklist Size]\label{inconspicuous-cor}
Under the conditions of \cref{inconspicuous},
the combined size of all blacklists in $\prefs{W}$ is at most $n\!-\!1$,
i.e.\ the average blacklist size is less than $1$.
\end{cor}

\cref{inconspicuous,inconspicuous-cor}  also have implications
in goods allocation problems, i.e.\ matching problems in which only one side (the \emph{buyers}) has
preferences, while the other side (the \emph{goods}) does not.
\cite{atila-school-choice} and \cite{atila-strategyproofness-efficiency} consider using a version of the (student-optimal)
Gale-Shapley algorithm in order to allocate school seats to children;
as school priorities are very coarse \citep[and sometimes nonexistent, see e.g.][]{nyc-match}, some tie-breaking rule is needed so that the algorithm can be run.
Both of these papers advocate the use of a single tie-breaking rule for all schools (e.g.\ a single lottery conducted once before the algorithm runs, to assign each student
a number to be used for tie breaking at every school) over a different tie-breaking rule for each school (e.g.\ a different lottery for each school), arguing
that the former results in higher social welfare. \cref{inconspicuous} strengthens this argument and makes it even more concrete, showing that in goods allocation problems
(in the absence of any priorities for the goods),
an implementation using multiple arbitrary (e.g.\ randomized) tie-breaking rules may yield (for the case of randomization, with possibly-small, albeit positive, probability)
any allocation that is buyer-rational,
regardless of the internal orders of the buyer's preference lists, if even an amortized less-than-one-buyer blacklist per good is allowed (and as we show in the next section, even if no blacklists are allowed whatsoever, if $|W|\!<\!|M|$); in contrast, it is easy to see that an implementation with a single tie-breaking
rule is equivalent to serial dictatorship, or to random serial dictatorship \citep{rsd} if the tie-breaking rule is randomized, and thus e.g.\ constructs Pareto-efficient allocations.

We note that since the combined size of the preference lists of $W$ as in \cref{inconspicuous}
is $\Theta(n^2)$, then as long as the algorithm for finding them must encode
its output explicitly,
its
time complexity
must be $\Omega(n^2)$. In \cref{flat-case} in the appendix,
we show that attaining $\Theta(n^2)$ time complexity is possible in the special
case of \cref{inconspicuous} in which the top choices of all men are distinct.
Somewhat surprisingly, the analysis and proof for this scenario, in which the men
attempt to force some matching as the $M$-optimal stable matching
\emph{\`a la} \cref{men-forcing}, 
are simpler than those of the general case,
as is calculating the women's ``response'' $\prefs{W}$.%
\footnote{%
Moreover, in this special case, while $\prefs{W}$
naturally depends on $\prefs{M}$,
if the implementation of the Gale-Shapley
algorithm that is used is the one proposed by \cite{Dubins-Freedman}, then
there exists a scheduling of this implementation for which the decision of whom each woman
prefers or blacklists on each step according to $\prefs{W}$ can be taken online.
More precisely, both the choice of who acts next, and the action of that participant if that participant is a woman,
depend solely the history of the run (and not on the not-yet-disclosed suffixes of the men's preference lists).
Thus, if participants are not required to submit their preference lists
in advance, but rather only to dynamically act upon them,
then a strategy for the women
that forces $\matching$ as the only stable matching against \emph{every profile} $\prefs{M}$ of
preference lists for $M$ can be constructed for \citeauthor{Dubins-Freedman}'s implementation of the algorithm, if the women can control its scheduling; this does not seem to be
possible in the general case of \cref{inconspicuous}.}
Thus, in a sense the men
actually inadvertently help the women whenever they try to manipulate the algorithm and force some matching as the
$M$-optimal stable matching.

We now show that \cref{inconspicuous}
is tight in a strong sense, w.r.t.\ the amount of blacklists and their sizes.
We conclude that \cref{inconspicuous} describes an optimal strategy for the
women in any model in which, all assigned partners being equal,
the utility of the coalition of
women decreases as the combined size of their blacklists increases (e.g.\
if women have to pay even a small cost for every man that they blacklist).
Indeed, in such models \cref{inconspicuous} beats
preference truncation by an order of magnitude.

\begin{thm}[Tightness of \cref{inconspicuous,inconspicuous-cor}]\label{inconspicuous-tight}
Let $n \in \naturals$ and
let $W$ and $M$ be sets of women and men, respectively,
s.t.\ $|W|\!=\!|M|\!=\!n$. For every $0 \le n_b \le \left\lfloor \frac{n}{2} \right\rfloor$ and for every
$l_1,\ldots,l_{n_b} > 0$ s.t.\ $l_1+\cdots+l_{n_b}\le n - n_b$,
there exist a profile $\prefs{M}$ of
preference lists for $M$ in which all blacklists are empty, and a perfect matching $\matching$,
s.t.\ both of the following hold.
\begin{parts}
\item\label{inconspicuous-tight-exists}
There exists a profile $\prefs{W}$
of preference lists for $W$ with precisely
$n_b$ nonempty blacklists, all pairwise disjoint and
of sizes $l_1,\ldots,l_{n_b}$,
s.t.\ the only stable matching,
given $\prefs{W}$ and $\prefs{M}$, is $\matching$.
\item\label{inconspicuous-tight-no-less}
For every profile $\prefs{W}'$
of preference lists for $W$
s.t.\
the only stable matching, given $\prefs{W}'$ and $\prefs{M}$, is $\matching$,
there exist $n_b$ distinct women $w_1,\ldots,w_{n_b} \in W$ s.t.\
$|B_{w_i}(\prefs{W}')|\ge l_i$ for every $0 \le i \le n_b$.
\end{parts}
\end{thm}

\begin{remark}\label{tightness-notes}
Regarding \crefpart{inconspicuous-tight}{no-less},
\begin{parts}
\item
As shown in our proof below,
the requirement that the
only stable matching, given $\prefs{W}'$ and $\prefs{M}$, is $\matching$,
may be replaced by the weaker assumption that the $M$-optimal stable matching,
given $\prefs{W}'$ and $\prefs{M}$, be $\matching$.
\item
No limitations are imposed regarding
the number of blacklists in $\prefs{W}'$ in which each man appears; in particular, the blacklists
in $\prefs{W}'$ need not necessarily be pairwise disjoint.
\end{parts}
\end{remark}

\subsubsection{Proof Outline of Theorem~\refintitle{inconspicuous}}

A full proof of \cref{inconspicuous} is given in the appendix.
The general flow of the proof of is as follows.  We construct a profile of preference lists for $W$ s.t.\ the top choice of every $w \in W$ is $\matching(w)$; thus,
the $W$-optimal stable matching is $\matching$, and it is enough to make sure that the \mbox{$M$-optimal} stable matching is $\matching$ as well. We construct this
profile iteratively.

Assume, for the time being, that the top choices of all men are distinct, i.e.\ each woman is serenaded-to on the first night of the Gale-Shapley algorithm by precisely
one man (in this case, in the absence
of any blacklists for women, the algorithm would stop on the first night).
We pick some woman $\tilde{w}$ who is not matched with $\matching(\tilde{w})$, and set her blacklist list so that she rejects the unique man $m$
serenading under her window on the first night.
We adjust the women's preferences (see \cref{build-base}), so that $m$ is continually rejected until he serenades under the window of $\matching(m)$, who then accepts him, and in turn
rejects the man $m'$ serenading under her window, who is continually rejected until he serenades under the window of $\matching(m')$, and so fourth until this \emph{rejection cycle} concludes with $\matching(\tilde{w})$ serenading under the window of $\tilde{w}$. We note that while $\tilde{w}$ may have rejected (and thus has blacklisted) more than one man during this rejection cycle,
all such men are matched with their $\matching$-partners at the end of this rejection cycle, and in addition so is $\matching(\tilde{w})$, whom no woman blacklisted; in fact, the only woman
who blacklists anyone so far is~$\tilde{w}$, and so we have that more men are now matched with their $\matching$-partners than have been blacklisted, so we are ``on the right track'', in a sense.

At this point, we would have na\"{\i}vely liked to pick another woman $\tilde{w}'$ who is unmatched with $\matching(\tilde{w}')$,
and initiate a similar rejection cycle triggered by her, and so fourth until $\matching$ is obtained \citep[arguing that this ``cycle-by-cycle'' simulation yields the $M$-optimal stable matching
as well, due to the outcome of the Gale-Shapley algorithm being invariant
under certain timing changes, an invariance established by][]{Dubins-Freedman}. Alas, it is quite possible that all ``nominees'' for the role of $\tilde{w}'$ have already rejected quite a few men during the previous
rejection cycle (that which was triggered by $\tilde{w}$). In this case, adjusting the preference list of~$\tilde{w}'$ to reject the single man serenading under her window entails having her blacklist
not only that man, but also every man she rejected in the previous rejection cycle in favour of this man, which would not yield the required low combined blacklist size (nor necessarily yield
disjoint blacklists). In this case, we take a step back and ``merge'' (see \cref{build-step}) what would have been the rejection cycle triggered by $\tilde{w}'$ into the previous rejection cycle
triggered by $\tilde{w}$, i.e.\ modify the preferences of $W$ without blacklisting an ``excessive'' amount of men,
so that the ``chain-reaction'' triggered by the rejection of $m$ by $\tilde{w}$, would cause not only all rejections from the cycle triggered by $\tilde{w}$ (as originally defined), but also the rejections from the cycle that
would have been triggered by~$\tilde{w}'$. Luckily, such successive ``merging'' can be done in an efficient manner, without the need to resimulate rejection cycles after each step.
Only when ``merging'' of additional rejection cycles is no longer possible (see \cref{build-cor}), do we start another rejection cycle (see induction step in the proof of \cref{flat-case}), knowing that any woman that we pick to trigger this new rejection cycle has not yet rejected any man in previous rejection cycles.

In the general case, if we let the Gale-Shapley algorithm run its course (using arbitrary preferences s.t.\ each woman $w$ prefers $\matching(w)$ over any other man),
then by the time the algorithm halts, it may already be the case that every woman has rejected quite a few men, and therefore, as in the previous case, cannot be used to trigger a rejection cycle
without blacklisting an ``excessive'' amount of men.
In this case, we show (see the induction step in the proof of \cref{inconspicuous}) that there exists at least one woman s.t.\ the rejection cycle that would have been triggered by her can be ``merged''
into the run of the algorithm before it has halted, without blacklisting an ``excessive'' amount of men. The analysis, and the corresponding modifications to the preferences of $W$, are considerably more delicate in the case. A side effect is that the time complexity is also somewhat higher (a worst case of $O(n^3)$ instead of $O(n^2)$), however certain properties of random permutations are used to show that the increment in the average-case time complexity is
considerably less significant (from $O(n^2)$ to $O(n^2\cdot\log n)$).

\subsection{Partial Matchings}

It turns out that the case in which $|W|\!=\!|M|$ is a singular case at which a phase change occurs. If there are more women than men (and therefore not all
women are matched in $\matching$), then each unmatched woman may be required to
blacklist as many as all men;
in contrast, if there are more men than women (and therefore not all
men are matched in $\matching$), and if the unmatched men do not blacklist any
woman,\footnote{A much weaker condition suffices. See e.g.\ \crefpart{cor-not-all-matched}{unmatched-men} below.}
then no blacklists are whatsoever required, turning the dichotomy between
\cref{men-forcing,overkill} on its head.
It is interesting to note
that \cite{Ashlagi-Leshno} have shown that a somewhat
similar phase change occurs w.r.t.\ the expected ranking of each participant's partner in the participant's
preference list in a random market. Both that paper and this one show, in different senses,
that the preferences of the smaller side of the market, even if it is only
slightly smaller, play a far more
significant role than may be expected in determining the stable matchings, and those of the larger side ---
a considerably insignificant one. In a sense, our results extend this qualitative statement from a random market to \emph{any} market.

\cref{inconspicuous-not-all-matched} formalizes the above
results, as well as their generalizations for partial matchings with
both unmatched women and unmatched men. \cref{tightness-not-all-matched} shows that these results are also tight in the same strong sense in which
\cref{inconspicuous-tight} shows that \cref{inconspicuous} is tight.
Before we formulate these theorems, we first define some notation.

\begin{defn}
Let $W$ and $M$ be sets of women and men, respectively.
For every $\tobematched{W} \subseteq W$ (resp.\ $\tobematched{M} \subseteq M$),
we define $\comp{\tobematched{W}}=W\setminus \tobematched{W}$ (resp.\ $\comp{\tobematched{M}}=M\setminus \tobematched{M}$).
(The set $W$ (resp.\ $M$) will be clear from context.)
\end{defn}

\begin{thm}[Manipulation with Minimal Blacklists to obtain a Partial Matching]\label{inconspicuous-not-all-matched}
Let $W$ and $M$ be sets of women and men of arbitrary sizes,
let $\prefs{M}$ be a profile of preference lists for $M$,
and let $\matching$ be a (possibly-partial) $M$-rational matching. 
Define $n_h \eqdef \bigl|\bigl\{w \in \matched{W} \mid \exists m \in \unmatched{M} : w \notin B_m(\prefs{M})\bigr\}\bigr|$.
There exists a profile $\prefs{W}$
of preference lists for $W$, s.t.\ all of the following hold.
\begin{parts}
\item
The only stable matching, given $\prefs{W}$ and $\prefs{M}$, is $\matching$.
\item
The blacklists of $\matched{W}$ in $\prefs{W}$ are pairwise disjoint, and contain only
members of $\matched{M}$.
\item
$n_b$, the number of women in $\matched{W}$ who have nonempty blacklists in $\prefs{W}$,
is at most $\left\lfloor \frac{n_{\matching}-n_h}{2} \right\rfloor$.
\item
The combined size of the blacklists of $\matched{W}$ in $\prefs{W}$
is at most $n_{\matching}\!-\!n_h\!-\!n_b$.
\end{parts}
Furthermore, $\prefs{W}$ can be computed in worst-case
$O\bigl(max\bigl\{|W|\cdot|M|,(n_{\matching}-n_h)\cdot {n_{\matching}}^2\bigr\}\bigr)$ time, best-case $O(|W|\cdot|M|)$ time and
average-case $O\bigl(max\bigl\{|W|\cdot|M|,{n_{\matching}}^2 \log \frac{n_{\matching}+1}{n_h+1}\bigr\}\bigr)$ time.
\end{thm}

\begin{cor}\label{cor-not-all-matched}
Under the conditions of \cref{inconspicuous},
\begin{parts}
\item
The combined size of all blacklists of $\matched{W}$
in $\prefs{W}$ is at most $n_{\matching}\!-\!n_h\!-\!1 \le n\!-\!1$,
i.e.\ the average blacklist size of $\matched{W}$ is less than $1$.
\item\label{cor-not-all-matched-unmatched-men}
If all women are matched in $\matching$ (i.e.\ $\matched{W}\!=\!W$),
and if for each $w\!\in\!W$ except for perhaps one,
there exists $m\!\in\!\unmatched{M}$ who does not
blacklist $w$ (e.g.\ this condition holds if at least one
$m\!\in\!\unmatched{M}$ blacklists at most one woman),
then all blacklists in $\prefs{W}$ are empty, and $\prefs{W}$ can be computed
in $O\bigl(|W|\cdot|M|\bigr)$ time, which constitutes a tight bound.
\end{parts}
\end{cor}

\begin{thm}[Tightness of \cref{inconspicuous-not-all-matched,cor-not-all-matched}]\label{tightness-not-all-matched}
Let $W$ and $M$ be sets of women and men of arbitrary sizes, and let
$\tobematched{W} \!\subseteq\! W$ and $\tobematched{M} \!\subseteq\! M$ s.t.\
$|\tobematched{W}|\!=\!|\tobematched{M}|$.
For each $w \in \comp{\tobematched{W}}$, let $B_w \!\subseteq\! M$, and
for each~$m \in \comp{\tobematched{M}}$, let $B_m \!\subseteq\! W$.
Define $n_h \eqdef \bigl|\bigl\{w \in \tobematched{W} \mid \exists m \in
\comp{\tobematched{M}} : w \notin B_m\bigr\}\bigr|$.
For every $0 \le n_b \le \left\lfloor \frac{n_{\matching}-n_h}{2} \right\rfloor$ and for every
$l_1,\ldots,l_{n_b} > 0$ s.t.\ $l_1+\cdots+l_{n_b}\le n_{\matching} - n_h - n_b$,
there exist a profile~$\prefs{M}$ of
preference lists for $M$ s.t.\
$B_m(\prefs{M})\!=\!B_m$ for each $m \in \comp{\tobematched{M}}$,
and a matching $\matching$ s.t.\ $\matched{W}\!=\!\tobematched{W}$ and $\matched{M}\!=\!\tobematched{M}$,
s.t.\ both of the following hold.
\begin{parts}
\item\label{tightness-not-all-matched-exists}
There exists a profile $\prefs{W}$
of preference lists for $W$, in which
$B_w(\prefs{W})\!\supseteq\! B_w$ for each $w \in \comp{\tobematched{W}}$,
and in which
$\tobematched{W}$ have precisely
$n_b$ nonempty blacklists, all pairwise disjoint and containing only members of $\tobematched{M}$,
and of sizes $l_1,\ldots,l_{n_b}$,
s.t.\
the only stable matching, given $\prefs{W}$ and $\prefs{M}$, is $\matching$.
\item\label{tightness-not-all-matched-no-less}
For every profile $\prefs{W}'$
of preference lists for $W$
s.t.\
the only stable matching, given $\prefs{W}'$ and $\prefs{M}$, is~$\matching$,
it holds that
$B_w(\prefs{W}')\supseteq \{m \in B_w(\prefs{W}) \mid w \notin B_m\}$
for each $w \in \comp{\tobematched{W}}$,
and in
addition, there exist $n_b$ distinct women $w_1,\ldots,w_{n_b} \in \tobematched{W}$ s.t.\
$|B_w(\prefs{W}') \cap \tobematched{M}| \ge l_i$ for every $0 \le i \le n_b$.
\end{parts}
\end{thm}

\begin{remark}
As in \cref{tightness-notes}, we note the following also regarding
\crefpart{tightness-not-all-matched}{no-less}.
\begin{parts}
\item
As shown in our proof below,
the requirement that the
only stable matching, given $\prefs{W}'$ and $\prefs{M}$, is $\matching$,
may be replaced by the weaker assumption that the $M$-optimal stable matching,
given $\prefs{W}'$ and $\prefs{M}$, be $\matching$.
\item
No limitations are imposed regarding
the number of blacklists in $\prefs{W}'$ in which each man appears; in particular, not even the
blacklists of $\tobematched{W}$ in $\prefs{W}'$ need necessarily be pairwise
disjoint.
\end{parts}
\end{remark}

\section{Divorces}\label{divorces}
In \cref{blacklists}, we showed that women can force
any ($M$-rational) matching
as the outcome of the Gale-Shapley algorithm, via
a small amount of nonempty blacklists, and average blacklist size less than one;
nonetheless, as shown in that section, some of these blacklists may be quite long, having length
in the order of magnitude of the number of men. We now show 
that, turning blacklisting into a dynamic process in the
Gale-Shapley algorithm, women can force any $M$-rational matching
as the outcome
by means of at most one blacklist-driven rejection per woman
(\emph{without} averaging),
but by a possibly-larger number of women than in \cref{blacklists}.

\begin{defn}[Gale-Shapley Algorithm with Divorces]
We extend the Gale-Shapley algorithm to allow for \emph{divorces}
as follows. The extended algorithm is divided into stages to which we refer as
\emph{seasons}. In the first season, the \emph{vanilla} (i.e.\ without divorces)
Gale-Shapley
algorithm runs until it converges. At the end of each season $s$, if any woman asks
to divorce her then-current partner, then one of these women, denoted $w$, is
chosen (arbitrarily, or according to some predefined rule); in
season $s\!+\!1$, the vanilla Gale-Shapley algorithm runs once more, starting at the
state concluding season $s$, with $w$ rejecting her then-current partner
on the first night (e.g.\ on the second night, this partner serenades
under the window of his next choice after $w$), and running until it converges once more.
When a season concludes with no woman asking for a divorce, the Gale-Shapley
algorithm with divorces concludes.
\end{defn}

\begin{remark}
If all women have a divorce strategy of ``never-divorce'', then the
Gale-Shapley algorithm with divorces concludes after one season, and is
thus equivalent to the vanilla Gale-Shapley algorithm.
\end{remark}

We note that we define the Gale-Shapley algorithm with divorces in terms
of seasons for ease of presentation; indeed,
all the results in this section hold verbatim, via conceptually-similar proofs, if we allow
every woman to divorce her partner instantaneously whenever she so desires.

Before we present our results regarding divorces, we first exhibit
the differences between blacklists and divorces on one hand,
and the similarities between these on the other hand. To exhibit
the qualitative distinction between these two concepts,
consider for a moment a woman~$w$ who wishes to blacklist
all men who approach her on the first night of the vanilla Gale-Shapley
algorithm; such a (potentially-long) blacklist can be ``replaced'' by a
single divorce in the following manner: $w$ places all such men at the
end of her preference list in arbitrary order, and thus rejects all such men
but one of them, denoted $m$, on the first night; if $w$ does not reject $m$
by the end of the season, then she asks to divorce him (possibly repeatedly,
until she is chosen to do so).\footnote{We note that if men were
allowed to divorce as well, then they could also use divorces to replace
blacklists, but to achieve a dual effect, i.e.\ instead of one divorce replacing
the entire blacklist of one man, one divorce could allow the blacklisting of
a woman by a single one man instead of by a set of men. Indeed, consider e.g.\ the
scenario in which all men wish to blacklist all women, forcing the empty
matching. In this case,
the men could all set the same order of preference for themselves;
they would thus all serenade under the window of the same woman $w$ on the
first night and $w$ would thus reject all of them but a single man $m$, who
would later divorce her. Similarly, every woman would have to be divorced
only by her most-favourite man, thus replacing blacklists with combined size
$n^2$, with $n$ divorces, each of a distinct woman but not necessarily
by a distinct man. We note that in this example, due to the dynamic nature
of divorces, this divorce
strategy for the men would force the empty matching as the outcome
against every profile
$\prefs{W}$ of preference lists and divorces for $W$, just as blacklisting
all women would.}
Indeed, as shown in \cref{one-divorce-per-woman,divorce-tight},
no woman need divorce more than one man in order
to achieve the lower bound on the required number of divorces;
nonetheless, as we show in \cref{divorce-tight}, the tight bound
of \cref{inconspicuous-cor} on the combined size of all blacklists cannot
be improved upon by replacing blacklists with divorces.
\cref{divorce-equiv-blacklist}, which is an immediate consequence
of the invariance of the outcome of the vanilla Gale-Shapley algorithm under timing changes~\citep{Dubins-Freedman}, exhibits an equivalence, in a sense, of the strengths
of blacklists and divorces.

\begin{prop}[Equivalent Strength of Divorces and Blacklists]\label{divorce-equiv-blacklist}
Let $W$ and $M$ be equal-sized sets of women and men, respectively,
let $\prefs{W}$ be a profile of preference lists and divorce strategies for $W$
and let $\prefs{M}$ be a profile of preference lists for $M$. Let $w \in W$.
\begin{parts}
\item
If $w$'s divorce strategy in $\prefs{W}$ is that of ``never-divorce'', then removing
$w$'s blacklist, placing its members in the end of $w$'s preference list in
arbitrary order, and replacing $w$'s divorce strategy with ``always ask
for a divorce if $w$'s current partner is in $B_w(\prefs{W})$'', does not
alter the outcome of the Gale-Shapley algorithm with divorces.
\item
If $w$ has an empty blacklist in $\prefs{W}$, then replacing her
divorce strategy with a ``never-divorce'' strategy, and blacklisting 
the minimal set $B$ s.t.\ $B$ contains all
the men $w$ divorces during the run of the Gale-Shapley algorithm with
divorces given $\prefs{W}$ and $\prefs{M}$ and s.t.\ every man that $w$ rejects
during this run in favour of a man in $B$ is also in $B$,\footnote{$B$ is the set of men reachable by means of the transitive closure of the
operation ``all men who are rejected by $w$ during this run in favour of \ldots'',
starting from a man divorced by $w$ during this run.}
does not alter the outcome of the Gale-Shapley algorithm with divorces.
\end{parts}
\end{prop}

We are now ready to formulate our results regarding the Gale-Shapley
algorithm with divorces.

\begin{thm}[One Divorce per Woman]\label{one-divorce-per-woman}
Let $W$ and $M$ be equal-sized sets of women and men, respectively. Define
$n\!\eqdef\!|W|\!=\!|M|$.
Let $\prefs{M}$ be a profile of preference lists for $M$.
For every $M$-rational perfect matching
$\matching$, there exist a profile $\prefs{W}$
of preference lists and divorce strategies for $W$,
s.t.\ all of the following hold.
\begin{parts}
\item
The run of the Gale-Shapley algorithm with divorces
according to $\prefs{W}$ and $\prefs{M}$, which we henceforth denote by $R$,
yields $\matching$.
\item
All blacklists in $\prefs{W}$ are empty.
\item
At most $n\!-\!1$ divorces occur during $R$,
each of them by a distinct woman. Moreover, 
if a season of $R$ commences with the divorce of a woman $w$,
then this season concludes with $w$ becoming matched with $\matching(w)$.
\end{parts}
\end{thm}

\begin{thm}[Tightness of \cref{one-divorce-per-woman}]\label{divorce-tight}
Let $n\!\in\!\naturals$ and
let $W$ and $M$ be sets of women and men, respectively,
s.t.\ $|W|\!=\!|M|\!=\!n$.
There exist a profile $\prefs{M}$ of
preference lists for $M$, and an $M$-rational perfect matching $\matching$, s.t.\
for every profile $\prefs{W}'$
of preference lists and divorce strategies for $W$ in which all blacklists are
empty
s.t.\
the run of the Gale-Shapley algorithm with divorces,
according to $\prefs{W}'$ and $\prefs{M}$, yields $\matching$, at least
$n\!-\!1$ divorces occur in this run.
\end{thm}

Similarly to
\cref{inconspicuous-not-all-matched,tightness-not-all-matched}, it is possible to show that in the case of partial matchings, no divorces (nor blacklists) are required by any woman who is not blacklisted
by at least one unmatched man.
While any unmatched woman, in the context of \cref{inconspicuous-not-all-matched}, may
be required to blacklist as many as all men, in many cases when divorces are allowed,
such blacklists of all men can each be replaced by a single divorce.

\section*{Acknowledgements}
This work was supported in part by an ISF grant, by the Google Interuniversity
Center for Electronic Markets and Auctions, and by
the European Research Council under the European Community's Seventh Framework
Programme (FP7/2007-2013) / ERC grant agreement no.\ [249159].
The author would like to thank Sergiu Hart and Noam Nisan, his Ph.D.\ advisors,
for useful discussions and comments; Assaf Romm for providing him with useful 
references to the literature; and Jacob Leshno for
suggesting the implication in goods allocation.

\bibliographystyle{abbrvnat}
\bibliography{blacklists}

\begin{thebibliography}{28}
\providecommand{\natexlab}[1]{#1}
\providecommand{\url}[1]{\texttt{#1}}
\expandafter\ifx\csname urlstyle\endcsname\relax
  \providecommand{\doi}[1]{doi: #1}\else
  \providecommand{\doi}{doi: \begingroup \urlstyle{rm}\Url}\fi

\bibitem[Abdulkadiro{\u{g}}lu and S{\"o}nmez(1998)]{rsd}
A.~Abdulkadiro{\u{g}}lu and T.~S{\"o}nmez.
\newblock Random serial dictatorship and the core from random endowments in
  house allocation problems.
\newblock \emph{Econometrica: Journal of the Econometric Society}, 66\penalty0
  (3):\penalty0 689--710, 1998.

\bibitem[Abdulkadiro{\u{g}}lu and S{\"o}nmez(2003)]{atila-school-choice}
A.~Abdulkadiro{\u{g}}lu and T.~S{\"o}nmez.
\newblock School choice: A mechanism design approach.
\newblock \emph{American Economic Review}, 93\penalty0 (3):\penalty0 729--747,
  2003.

\bibitem[Abdulkadiro{\u{g}}lu et~al.(2005{\natexlab{a}})Abdulkadiro{\u{g}}lu,
  Pathak, and Roth]{nyc-match}
A.~Abdulkadiro{\u{g}}lu, P.~A. Pathak, and A.~E. Roth.
\newblock The new york city high school match.
\newblock \emph{American Economic Review}, 95\penalty0 (2):\penalty0 364--367,
  2005{\natexlab{a}}.

\bibitem[Abdulkadiro{\u{g}}lu et~al.(2005{\natexlab{b}})Abdulkadiro{\u{g}}lu,
  Pathak, Roth, and S{\"o}nmez]{boston-match}
A.~Abdulkadiro{\u{g}}lu, P.~A. Pathak, A.~E. Roth, and T.~S{\"o}nmez.
\newblock The boston public school match.
\newblock \emph{American Economic Review}, 95\penalty0 (2):\penalty0 368--371,
  2005{\natexlab{b}}.

\bibitem[Abdulkadiro{\u{g}}lu et~al.(2009)Abdulkadiro{\u{g}}lu, Pathak, and
  Roth]{atila-strategyproofness-efficiency}
A.~Abdulkadiro{\u{g}}lu, P.~A. Pathak, and A.~E. Roth.
\newblock Strategy-proofness versus efficiency in matching with indifferences:
  Redesigning the nyc high school match.
\newblock \emph{American Economic Review}, 99\penalty0 (5):\penalty0
  1954--1978, 2009.

\bibitem[Arratia et~al.(2003)Arratia, Barbour, and Tavar{\'e}]{Arratia-et-al}
R.~Arratia, A.~Barbour, and S.~Tavar{\'e}.
\newblock \emph{Logarithmic Combinatorial Structures: a Probabilistic
  Approach}.
\newblock EMS Monographs in Mathematics. European Mathematical Society, Zurich,
  Switzerland, 2003.

\bibitem[Ashlagi et~al.(2013)Ashlagi, Kanoria, and Leshno]{Ashlagi-Leshno}
I.~Ashlagi, Y.~Kanoria, and J.~D. Leshno.
\newblock Unbalanced random matching markets.
\newblock In \emph{Proceedings of the 14th ACM Conference on Electronic
  Commerce (EC)}, pages 27--28, 2013.

\bibitem[Demange et~al.(1987)Demange, Gale, and Sotomayor]{DGS87}
G.~Demange, D.~Gale, and M.~Sotomayor.
\newblock A further note on the stable matching problem.
\newblock \emph{Discrete Applied Mathematics}, 16\penalty0 (3):\penalty0
  217--222, 1987.

\bibitem[Dubins and Freedman(1981)]{Dubins-Freedman}
L.~E. Dubins and D.~Freedman.
\newblock Machiavelli and the {G}ale-{S}hapley algorithm.
\newblock \emph{American Mathematical Monthly}, 88\penalty0 (7):\penalty0
  485--494, 1981.

\bibitem[Eeckhout(2000)]{Eeckhout}
J.~Eeckhout.
\newblock On the uniqueness of stable marriage matchings.
\newblock \emph{Economics Letters}, 69\penalty0 (1):\penalty0 1--8, 2000.

\bibitem[Gale and Shapley(1962)]{Gale-Shapley}
D.~Gale and L.~S. Shapley.
\newblock College admissions and the stability of marriage.
\newblock \emph{American Mathematical Monthly}, 69\penalty0 (1):\penalty0
  9--15, 1962.

\bibitem[Gale and Sotomayor(1985)]{Gale-Sotomayor-ms-machiavelli}
D.~Gale and M.~Sotomayor.
\newblock Ms. {M}achiavelli and the stable matching problem.
\newblock \emph{American Mathematical Monthly}, 92\penalty0 (4):\penalty0
  261--268, 1985.

\bibitem[Gonczarowski(2013)]{Combined-Length-Distribution}
Y.~A. Gonczarowski.
\newblock The distribution of the combined length of spanned cycles in a random
  permutation.
\newblock Discussion Paper 650, Center for the Study of Rationality, Hebrew
  University of Jerusalem, 2013.

\bibitem[Gonczarowski and Friedgut(2013)]{Sisterhood}
Y.~A. Gonczarowski and E.~Friedgut.
\newblock Sisterhood in the {G}ale-{S}hapley matching algorithm.
\newblock \emph{The Electronic Journal of Combinatorics}, 20\penalty0
  (2):\penalty0 \#P12 (18pp), 2013.

\bibitem[Gusfield and Irving(1989)]{GI89}
D.~Gusfield and R.~W. Irving.
\newblock \emph{The Stable Marriage Problem: Structure and Algorithms}.
\newblock The MIT Press, Cambridge, MA, USA, 1989.

\bibitem[Huang(2006)]{H06}
C.-C. Huang.
\newblock Cheating by men in the {G}ale-{S}hapley stable matching algorithm.
\newblock In \emph{Proceedings of the 14th Annual European Symposium on
  Algorithms (ESA)}, pages 418--341, 2006.

\bibitem[Jensen(1906)]{Jensen}
J.~L. W.~V. Jensen.
\newblock Sur les fonctions convexes et les in\'{e}galit\'{e}s entre les
  valeurs moyennes.
\newblock \emph{Acta Mathematica}, 30\penalty0 (1):\penalty0 175--193, 1906.

\bibitem[Kobayashi and Matsui(2009)]{KobayashiMatsui}
H.~Kobayashi and T.~Matsui.
\newblock Successful manipulation in stable marriage model with complete
  preference lists.
\newblock \emph{IEICE Transactions on Information and Systems}, E92.D\penalty0
  (2):\penalty0 116--119, 2009.

\bibitem[Kojima and Pathak(2009)]{KojimaPathak}
F.~Kojima and P.~A. Pathak.
\newblock Incentives and stability in large two-sided matching markets.
\newblock \emph{American Economic Review}, 99\penalty0 (3):\penalty0 608--627,
  2009.

\bibitem[McVitie and Wilson(1970)]{MW70}
D.~G. McVitie and L.~B. Wilson.
\newblock Stable marriage assignment for unequal sets.
\newblock \emph{BIT}, 10\penalty0 (3):\penalty0 295--309, 1970.

\bibitem[McVitie and Wilson(1971)]{MW71}
D.~G. McVitie and L.~B. Wilson.
\newblock The stable marriage problem.
\newblock \emph{Communications of the ACM}, 14\penalty0 (7):\penalty0 486--490,
  1971.

\bibitem[Pini et~al.(2011)Pini, Rossi, Venable, and Walsh]{PRVWAAMAS2011}
M.~S. Pini, F.~Rossi, K.~B. Venable, and T.~Walsh.
\newblock Manipulation complexity and gender neutrality in stable marriage
  procedures.
\newblock \emph{Autonomous Agents and Multi-Agent Systems}, 22\penalty0
  (1):\penalty0 183--199, 2011.

\bibitem[Pittel(1989)]{PittelAverageNumber}
B.~Pittel.
\newblock The average number of stable matchings.
\newblock \emph{SIAM Journal on Discrete Mathematics}, 2\penalty0 (4):\penalty0
  530--549, 1989.

\bibitem[Pittel(1992)]{PittelLikelySolutions}
B.~Pittel.
\newblock On likely solutions of a stable marriage problem.
\newblock \emph{The Annals of Applied Probability}, 2\penalty0 (2):\penalty0
  358--401, 1992.

\bibitem[Roth(1984)]{Roth-labor-evolution}
A.~E. Roth.
\newblock The evolution of the labor market for medical interns and residents:
  a case study in game theory.
\newblock \emph{Journal of Political Economy}, 92\penalty0 (6):\penalty0
  991--1016, 1984.

\bibitem[Roth(2002)]{roth-economist-engineer}
A.~E. Roth.
\newblock The economist as engineer: Game theory, experimentation and
  computation as tools for design economics.
\newblock \emph{Econometrica}, 70\penalty0 (4):\penalty0 1341--1378, 2002.

\bibitem[Roth et~al.(2005)Roth, S{\"o}nmez, and {\"U}nver]{kidneys}
A.~E. Roth, T.~S{\"o}nmez, and M.~U. {\"U}nver.
\newblock A kidney exchange clearinghouse in new england.
\newblock \emph{American Economic Review}, 95\penalty0 (2):\penalty0 376--380,
  2005.

\bibitem[Teo et~al.(2001)Teo, Sethuraman, and Tan]{TST01}
C.-P. Teo, J.~Sethuraman, and W.-P. Tan.
\newblock {G}ale-{S}hapley stable marriage problem revisited: Strategic issues
  and applications.
\newblock \emph{Management Sciences}, 47\penalty0 (9):\penalty0 1252--1267,
  2001.

\end{thebibliography}

\clearpage
\appendix
\section{Proofs}\label{proofs}

\subsection{Proof of Theorem~\refintitle{inconspicuous}}

Let $W$, $M$ and $n$ be as in \cref{inconspicuous}.
We begin by defining some notation.

\begin{defn}
Let $\prefs{W}$ and $\prefs{M}$ be
profiles of preferences lists for $W$ and for $M$, respectively.
\begin{itemize}
\item
We denote by $\fullrun$ the run
of the Gale-Shapley algorithm according to $\prefs{W}$ and $\prefs{M}$,
and denote the matching that $\fullrun$ yields (i.e.\ the $M$-optimal
stable matching for $\prefs{W}$ and $\prefs{M}$) by $\menopt(\prefs{W},\prefs{M})$.
\item
Let $\matching'$ be a matching.
We denote by $\run$ the run
of the Gale-Shapley algorithm according to $\prefs{W}$ and $\prefs{M}$,
starting with $\matching'$ as the initial state, and denote the
matching that $\run$ yields by $\resultmatching$.
\end{itemize}
\end{defn}

At the heart of our proof of \cref{inconspicuous}
lies a combinatorial structure that we call
a $\emph{cycle}$, and which we now define.

\begin{defn}\label{cycle}
\leavevmode
\begin{parts}
\item
$\generalcycle$, for $d \in \mathbb{N}$, is called a \emph{cycle}
if for every $i < d$, $w_i \in W$ and $m_i \in M$, and if $w_d=w_1$.
\item
We say that a cycle $C\eqdef\generalcycle$ is \emph{simple} if $w_1,\ldots,w_{d-1}$, $m_1,\ldots,m_{d-1}$ are all distinct,
i.e.\ if the only participant appearing in $C$ more than once
is $w_1$, appearing both as $w_1$ and as $w_d$.
\item
We say that two cycles $C,\tilde{C}$ are
are \emph{disjoint} if no participant appears in both.
\item
We say that a cycle $\tilde{C}$ is a \emph{cyclic shift} of a cycle $C\eqdef\generalcycle$
if there exists $1 \le \ell \le d$ s.t.\
$\tilde{C}=(w_{\ell} \xrightarrow{m_{\ell}} w_{\ell+1} \cdots
\xrightarrow{m_{d-1}} w_d=w_1 \xrightarrow{m_1} w_2 \xrightarrow{m_2}
\cdots w_{\ell})$.
\end{parts}
\end{defn}

Cycles can be used to naturally describe the dynamics following
a rejection in the Gale-Shapley algorithm, as we now show.

\begin{defn}\label{cycle-generating}
Let $\matching'$ be a matching.
We say that profiles $\prefs{W}$ and $\prefs{M}$
of preferences lists for $W$ and $M$, respectively, are \emph{$\matching'$-cycle generating}
if all of the following hold.
\begin{parts}
\item
For every $m \in M$, there exists a woman $w_m$ s.t.\
all of the following hold.
\begin{parts}
\item
$m$ is $w_m$'s top choice.
\item
$m$ does not blacklist $w_m$.
\item
$m$ weakly prefers $\matching'(m)$ over $w_m$.
\end{parts}
\item
There exists a unique woman $\tilde{w} \in W$ s.t.\ $\tilde{w}$ blacklists
$\matching'(\tilde{w})$. We call $\tilde{w}$ the \emph{cycle trigger}.
\end{parts}
\end{defn}

\begin{claim}\label{all-matched}
Let $\matching'$ be a matching.
For every $\matching'$-cycle-generating profiles
$\prefs{W}$ and $\prefs{M}$ of preferences lists for $W$ and $M$, respectively,
all of $W$ and $M$ are matched according to $\resultmatching$.
\end{claim}

\begin{proof}
As $|W|=|M|$, it is enough to show that every $m \in M$ is matched according to
$\resultmatching$.
Assume for contradiction that $m \in M$ is unmatched according to
$\resultmatching$.
As $w_m$ is not blacklisted by $m$, and as $m$ weakly prefers $\matching'(m)$
over $w_m$, we have, by $m$ being unmatched in $\resultmatching$, that $m$
serenades under $w_m$'s window during $\run$,
and that $w_m$ rejects $w$ during this run; but $m$ is $w_m$'s top choice
--- a contradiction.
\end{proof}

\begin{defn}\label{reject-cycle}
Let $\matching'$ be a matching.
Let $\prefs{W}$ and $\prefs{M}$ be $\matching'$-cycle-generating profiles
of preferences lists for $W$ and $M$, respectively,
with cycle trigger $\tilde{w}$.
We define the \emph{rejection cycle $(\prefs{W},\prefs{M})$-generated from $\matching'$},
denoted by $\rejectcycle$, as
$\generalcycle$, where $d$, $(w_i)_{i=1}^d$ and $(m_i)_{i=1}^d$ are defined
by following $\run$:
\begin{enumerate}
\item
$w_1 = \tilde{w}$, $m_1 = \matching'(\tilde{w})$. ($w_1$ rejects $m_1$ on the first night.)
\item
On the nights following the rejection of $m_i$ by $w_i$, $m_i$ serenades under
the windows of the women following $w_i$ in his preferences list. Denote
the first of these women to provisionally accept him by $w_{i+1}$.
If $w_{i+1}$ does not reject any man on the night on which she first provisionally
accepted $m_i$, then denote $d=i+1$ and the algorithm stops.
Otherwise, denote the unique
man rejected by $w_{i+1}$ on that night by $m_{i+1}$.
\end{enumerate}
\end{defn}

\begin{claim}
Under the conditions of \cref{reject-cycle},
$\rejectcycle$ is a well-defined cycle.
\end{claim}

\begin{proof}
We show by induction that on every night of the algorithm, one of the following
holds: (The claim follows directly from this proof.)
\begin{types}
\item\label{reject-cycle-night-one-to-one}
Each man serenades under a distinct woman's window. In this case, either
the algorithm stops or only $\tilde{w}$ rejects the man serenading under her
window.
\item\label{reject-cycle-night-empty-tilde-w}
No man serenades under $\tilde{w}$'s window, two men serenade under the
window of some other woman $w$, and each other woman is serenaded-to by
one man. In this case, one of the men serenading under $w$'s window is the
only man rejected on this night.
\end{types}

Base: By definition, on the first night each man $m \in M$ serenades under
$\matching'(m)$'s window, and all of these women are distinct.
By definition of $\tilde{w}$, indeed the only rejection on the first night
is of $m_1$ by $w_1$. Therefore, the first night is of \cref{reject-cycle-night-one-to-one}.

Step: Assume that the algorithm does not stop at the end of the $i$'th night.
Regardless of the type of the $i$'th night, by the induction hypothesis
exactly one man $m$ is rejected
during it, while $n-1$ men are provisionally accepted during it, each by a distinct
woman from $W \setminus \{\tilde{w}\}$.
By \cref{all-matched}, on the $i+1$'th night $m$ serenades under some
woman $w$'s window. If $w=\tilde{w}$, we show that the $i+1$'th
night is of \cref{reject-cycle-night-one-to-one}. Indeed, in this case each woman is serenaded-to
by exactly one man. Furthermore, each women $w' \in W \setminus \{\tilde{w}\}$ is
serenaded-to by the man she provisionally accepts on the $i$'th night;
therefore, this man is not in $w'$'s blacklist, and being the only
man serenading under $w'$'s window on the $i+1$'th night, he is not
rejected on this night. Thus, no man but $m$ is rejected the $i+1$'th night.
If $m$ is not rejected on this night, then the algorithm stops; otherwise,
$m$ is rejected by the woman under whose window he serenades, that is,
by $\tilde{w}$.

Otherwise, $w \ne \tilde{w}$ and in this case we show that
the $i+1$'th night is of \cref{reject-cycle-night-empty-tilde-w}. Indeed, in this case, on the $i+1$'th night no one serenades
under $\tilde{w}$'s window, two men serenades under $w$'s window ($m$
and the man $w$ provisionally accepts on the $i$'th night, denoted henceforth
as $m'$), and each
woman $W \setminus \{\tilde{w},w\}$ is serenaded-to by exactly one man.
Similarly to the previous case, each women in $W \setminus \{w, \tilde{w}\}$ is 
serenaded-to by the man she provisionally accepts on the $i$'th night;
therefore, this man is not in this woman's blacklist, and being the only
man serenading under her window on the $i+1$'th night, he is not
rejected on this night.
Thus, no woman but $w$ rejects any man during this night.
As $m'$ is provisionally accepted by $w$ on the $i$'th night, he is not
blacklisted by her, and thus she does not reject both $m$ and $m'$
on the $i+1$'th night. Thus, she rejects exactly one of them on this night
and the proof is complete.
\end{proof}

\begin{claim}\label{changed-if-reject-cycle}
Under the conditions of \cref{reject-cycle},
the participants for whom $\resultmatching(p) \ne \matching'(p)$ are exactly
the participants in $\rejectcycle$.
\end{claim}

\begin{proof}
By \cref{reject-cycle}, the participants in $\resultmatching(p)$ are exactly the participants who reject
or are rejected during $\run$.
\end{proof}

We now move on to define, given two matchings, a cycle associated with both.
While this definition is syntactic, and \emph{a priori} not related to any run,
in \cref{build-base,build-step} we show its relation to \cref{reject-cycle}.

\begin{defn}\label{abstract-cycle}
Let $\matching'$ and $\matching$ be two matchings and let $w \in W$.
We define a sequence $S_{\matching'}^{\matching}(w)\eqdef\bigl((w_i,m_i)\bigr)_{i=1}^{\infty} \in (W \times M)^{\naturals}$ as follows:
\begin{enumerate}
\item
$w_1 = w$.
\item
$m_i = \matching'(w_i)$.
\item
$w_{i+1} = \matching(m_i)$.
\end{enumerate}
Let $d>1$ be minimal s.t.\ $w_d = w_1$. If $d>2$, then
we define the \emph{$(\matching'\rightarrow\matching)$-cycle} of $w$,
denoted by $\abstractcycle{w}$, as $\generalcycle$. Otherwise
(i.e.\ if $\matching(w)=\matching'(w)$),
we define $\abstractcycle{w}=(w)$.
\end{defn}

\begin{claim}\label{abstract-cycle-properties}
\leavevmode
\begin{parts}
\item\label{abstract-cycle-properties-simple}
$\abstractcycle{w}$ is a well-defined simple cycle, for every $w \in W$.
\item\label{abstract-cycle-properties-disjoint-or-rotate}
For every $w,w' \in W$, $\abstractcycle{w}$ and
$\abstractcycle{w'}$ are either disjoint, or cyclic shifts of one another.
\item\label{abstract-cycle-properties-partition}
$W$ is partitioned into (the women from) $(\matching'\rightarrow\matching)$-cycles.
\end{parts}
\end{claim}

\begin{proof}
Let $w \in W$ and denote 
$\bigl((w_n,m_n)\bigr)_{i=1}^{\infty}\eqdef S_{\matching'}^{\matching}(w)$.
By definitions of $\matching$ and $\matching'$, for each $i$, $w_i$
uniquely determines $m_i$ and (if $i \ne 1$) $w_{i-1}$; similarly,
$m_i$ uniquely determines $w_i$ and $w_{i+1}$.
Thus, $S_{\matching'}^{\matching}(w)$ is periodic, and a single period of
this sequence contains no participant more than once; therefore,
\cref{abstract-cycle-properties-simple} follows.
We furthermore conclude that for every $w,w' \in W$, $S_{\matching'}^{\matching}(w)$ and $S_{\matching'}^{\matching}(w')$
are either disjoint, or each a suffix of the other,
and thus \cref{abstract-cycle-properties-disjoint-or-rotate} follows.
Finally, \cref{abstract-cycle-properties-partition} is an immediate consequence of \cref{abstract-cycle-properties-disjoint-or-rotate}.
\end{proof}

\begin{remark}\label{cycles-perm}
There is is a natural isomorphisms between
$(\matching'\rightarrow\matching)$-cycles modulo cyclic shifts,
and cycles of the permutation $\matching \circ {\matching'}^{-1}$
(resp.\ of the permutation $\matching' \circ \matching^{-1}$), given by taking
all women (resp.\ all men; or $\matching(w)$, if $d=1$) of a $(\matching'\rightarrow\matching)$-cycle
in order.
\end{remark}

Our proof strategy for \cref{inconspicuous}, as
seen below, is to iteratively modify the preferences of $W$, so that the
$M$-optimal matching they yield becomes ``closer and closer'', in a sense,
to the target matching $\matching$. We now define the
property required of these intermediate $M$-optimal matchings $\matching'$
w.r.t.\ the preferences of $M$, for this process to be able to continue
until $\matching$ is achieved.

\begin{defn}
Let $\matching'$ and $\matching$ be matchings.
We say that a profile of preference lists for $M$ is
\emph{$(\matching'\rightarrow\matching)$-compatible} if $m\in M$
weakly prefers $\matching'(m)$ over $\matching(m)$, but does not blacklist
$\matching(m)$.
\end{defn}

\begin{claim}\label{still-compatible}
Let $\matching'$ and $\matching$ be matchings and let $\prefs{M}$ be an
$(\matching'\rightarrow\matching)$-compatible profile of preference
lists for $M$. For every profile $\prefs{W}$ of preference lists for
$W$ according to which each woman $w$'s top choice is $\matching(w)$,
all of $W$ and $M$ are matched according to $\resultmatching$, and
$\prefs{M}$ is also
$(\resultmatching\rightarrow\matching)$-compatible.
\end{claim}

\begin{proof}
The proof is similar to that of \cref{all-matched}.
Let $m \in M$. We need only show that $m$ is matched in $\resultmatching(m)$
and weakly prefers $\resultmatching(m)$
over $\matching(m)$. Since $m$ weakly prefers $\matching'(m)$ over $\matching(m)$,
whom he does not blacklist, we have that $m$ would serenade under $\matching(m)$'s window during $\run$ before serenading under the window of any woman he
prefers less, and before becoming unmatched. As
$\matching(m)$'s top choice is $m$, she would never reject him, and the proof
is complete.
\end{proof}

We now begin the journey toward developing \cref{build-cor}, which is the main
technical tool underlying our proof of
\cref{inconspicuous}. \cref{build-cor} is proven
inductively, by continuously tweaking the preferences of $W$ until
they meet a certain property. \cref{pm-compatible} defines the
induction-invariant properties of the preferences of $W$ during this process,
while \cref{build-base,build-step} define the induction base and step of this process,
respectively.

\begin{defn}\label{pm-compatible}
Let $\matching'$ and $\matching$ be matchings and let $\prefs{M}$ be an
$(\matching'\rightarrow\matching)$-compatible profile of preference
lists for $M$. We say that a profile $\prefs{W}$ of preference lists for $W$
is \emph{$(\matching'\xrightarrow{\prefs{M}}\matching)$-compatible} if all the
following hold.
\begin{parts}
\item\label{pm-compatible-cycle-generating}
$\prefs{W}$ and $\prefs{M}$ are $\matching'$-cycle generating.
\item\label{pm-compatible-idol}
Each woman $w$'s top choice is $\matching(w)$.
\item\label{pm-compatible-trigger}
The cycle trigger for $\rejectcycle$ blacklists (at least)
every $m \in M$ s.t.\ $\resultmatching(m)\ne\matching(m)$.
\item\label{pm-compatible-non-trigger}
For every woman $w$ other than the cycle trigger, every $m\in M$ that $w$ prefers over $\matching'(w)$,
except perhaps for $\matching(w)$,
satisfies $\resultmatching(m)=\matching(m)$.
\item\label{pm-compatible-stay-or-opt}
For every $w \in W$, $\resultmatching(w) \in \bigl\{\matching'(w),
\matching(w) \bigr\}$.
\end{parts}
\end{defn}

\begin{claim}\label{pm-compatible-properties}
Under the conditions of \cref{pm-compatible}, all of the following hold.
\begin{parts}
\item\label{pm-compatible-properties-only-harbor}
Let $m \in M$ s.t.\ $\resultmatching(m)\ne\matching(m)$. The only $w \in W$
who prefers $m$ over $\matching'(w)$ is $\matching(m)$.
\item\label{pm-compatible-properties-stay-or-opt}
For every $p \in W \cup M$, $\resultmatching(p) \in \bigl\{\matching'(p),
\matching(p) \bigr\}$.
\item\label{pm-compatible-properties-opt-if-reject-cycle}
$\resultmatching(p)=\matching(p)\ne\matching'(p)$ for every participant $p$ in $\rejectcycle$.
\item\label{pm-compatible-properties-disjoint-union}
$\bigl\{p \in W \cup M \mid \resultmatching(p)=\matching(p)\bigr\}$ is a
(disjoint) union of
$(\matching'\rightarrow\matching)$-cycles.
\end{parts}
\end{claim}

\begin{proof}
\leavevmode
\begin{parts}
\item
Follows from \crefpart{pm-compatible}{non-trigger}.
\item
Follows from \crefpart{pm-compatible}{stay-or-opt}.
\item
By \cref{changed-if-reject-cycle}, every participant $p$ in $\rejectcycle$
satisfies $\resultmatching(p) \ne \matching'(p)$.
Thus, by \cref{pm-compatible-properties-stay-or-opt}, we also have
$\resultmatching(p) = \matching(p)$.
\item
Let $w$ be a woman s.t.\ $\resultmatching(w)=\matching(w)$.
If $\matching'(w)=\matching(w)$, then $\abstractcycle{w}=(w)$.
Otherwise, denote $\abstractcycle{w}=\generalcycle$, and we show by induction
that for all participants $p$ in $\abstractcycle{w}$,
$\resultmatching(p)=\matching(p)\ne\matching'(p)$.

Base: By definition, $w_1=w$, and thus $\resultmatching(w_1)=\matching(w_1)$
and by assumption, $\matching(w_1)\ne\matching'(w_1)$.

Step 1: Assume that $\resultmatching(w_i)=\matching(w_i)\ne\matching'(w_i)$
for some $1\le i < d$. Recall that $m_i=\matching'(w_i)$.
Since $m_i=\matching'(w_i) \ne \matching(w_i)$, we have
$\matching'(m_i)=w_i \ne \matching(m_i)$.
Since $m_i=\matching'(w_i) \ne \resultmatching(w_i)$, we have
$w_i=\matching'(m_i) \ne \resultmatching(m_i)$. Thus, by
\crefpart{pm-compatible}{stay-or-opt},
$\resultmatching(m_i)=\matching(m_i)$.

Step 2: Assume that $\resultmatching(m_i)=\matching(m_i)\ne\matching'(m_i)$
for some $1\le i < d$. Recall that $w_{i+1}=\matching(m_i)$.
Since $w_{i+1}=\matching(m_i)=\resultmatching(m_i)$, we have
$\matching(w_{i+1})=m_i=\resultmatching(w_{i+1})$.
Since $w_{i+1}=\matching(m_i) \ne \matching'(m_i)$, we have
$\matching(w_{i+1})=m_i \ne \matching'(w_{i+1})$.

The proof that for every $m \in M$ s.t.\ $\resultmatching(m)=\matching(m)\ne\matching'(m)$,
we have
$\resultmatching(p)=\matching(p)\ne\matching'(p)$
for all participants $p$ in $\abstractcycle{\matching(m)}$ (this cycle includes
$m$ as $m_{d-1}$), uses Step 2 as the induction base and Steps 1 and 2
as the induction steps.\qedhere
\end{parts}
\end{proof}

\begin{lemma}\label{build-base}
Let $\matching'$ and $\matching$ be matchings, let $\tilde{w} \in W$
s.t.\ $\matching(\tilde{w}) \ne \matching'(\tilde{w})$,
and let $\prefs{W}$ be a profile
of preference lists for $W$ according to which each woman $w \ne \tilde{w}$ has preference list
starting with $\matching(w)$, followed immediately by $\matching'(w)$
(if $\matching'(w)\ne\matching(w)$), followed by some or all
other men in arbitrary order, and in which $\tilde{w}$ has preference list
$\matching(\tilde{w})$, with all other men blacklisted.
For every $(\matching'\rightarrow\matching)$-compatible profile $\prefs{M}$
of preference lists for $M$, both of the following hold.
\begin{parts}
\item\label{build-base-compatible}
$\prefs{W}$ is $(\matching'\xrightarrow{\prefs{M}}\matching)$-compatible.
\item\label{build-base-cycle}
$\rejectcycle = \abstractcycle{\tilde{w}}$.
\end{parts}
\end{lemma}

\begin{proof}
It is straightforward to check, immediately from the definition
of $\prefs{W}$, that $\prefs{W}$ and $\prefs{M}$ are $\matching'$-cycle
generating with cycle trigger $\tilde{w}$ and with $w_m=\matching(m)$ for
every $m \in M$. Thus, in particular
\crefshowpart{pm-compatible}{cycle-generating} holds.
By definition of $\prefs{W}$, we trivially have that 
\crefshowpart{pm-compatible}{idol} holds.

By \cref{all-matched} and by definition of $\prefs{W}$, we have
$\resultmatching(\tilde{w})=\matching(\tilde{w})$.
Thus, $\resultmatching(\matching(\tilde{w}))=\matching(\matching(\tilde{w}))$,
and so,
by definition of $\prefs{W}$ we have that $\tilde{w}$ blacklists
every $m \in M$ s.t.\ $\resultmatching(m) \ne \matching(m)$ and
\crefshowpart{pm-compatible}{trigger} holds.

Let $w \in W \setminus \{\tilde{w}\}$. By definition of $\prefs{W}$, the only
man that $w$ possibly prefers over $\matching'(w)$ is $\matching(w)$ (and this
only happens if $\matching(w) \ne \matching'(w)$), and thus
\crefshowpart{pm-compatible}{non-trigger} vacuously holds.

For every woman $w \in W$, by a well-known property of the Gale-Shapley algorithm, we have that
$\resultmatching(w)$ is $w$'s most-preferred choice out of all
the men that serenade under her window during $\run$.
Thus, since on the first night of this run,
each woman $w \in W \setminus \{\tilde{w}\}$ is serenaded-to by $\matching'(w)$,
she weakly prefers $\resultmatching(w)$ over $\matching'(w)$; therefore,
by definition
of $\prefs{W}$, we obtain that $\resultmatching(w) \in \{\matching'(w),\matching(w)\}$ for every $w \in W \setminus \{\tilde{w}\}$.
Recall that $\resultmatching(\tilde{w})=\matching(\tilde{w})$. By both of these,
\crefshowpart{pm-compatible}{stay-or-opt} holds and the
proof of \cref{build-base-compatible} is complete.

We now prove \cref{build-base-cycle}.
Define $\rejectcycle=\generalcycle$. We show by induction that
$\generalcycle=\abstractcycle{\tilde{w}}$.

Base: As $\tilde{w}$ is the cycle trigger for $\rejectcycle$, we have,
by \cref{reject-cycle} that $w_1=\tilde{w}$ and $m_1=\matching'(\tilde{w})$,
agreeing with the definition of $\abstractcycle{\tilde{w}}$.

Step: Let $1 < i \le d$ and assume that $w_j$ and $m_j$ agree with the
definition of $\abstractcycle{\tilde{w}}$ for every $j \le i$.
We show
that $w_i$ and, if $i < d$, also $m_i$, agrees with the definition of $\abstractcycle{\tilde{w}}$ as well.
Let $t$ be the night on which $w_i$ rejects $m_i$.
On the first night, every $w \in W \setminus \{\tilde{w}\}$ is serenaded-to
by $\matching'(w)$, whom she does not blacklist;
therefore, once more by the same well-known property of the
Gale-Shapley algorithm as above, on every night there is a man serenading
under $w$'s window that she weakly prefers over $\matching'(w)$.
Thus, by definition of $\prefs{W}$, every
$w \in W \setminus \{\tilde{w}, \matching(m_i), \matching'(m_i) \}$ is serenaded-to
on every night by a man she strictly prefers over $m_i$, and would thus
reject $m_i$ if he ever serenaded under $w$'s window. Furthermore,
as, by the induction hypothesis, $w_i=\matching'(m_i)$
rejects $m_i$ on night $t$, we have that he does not serenade under her
window on any later night.
Thus, $w_{i+1} \in \{ \tilde{w}, \matching(m_i) \}$.
If $\tilde{w} = \matching(m_i)$, then $w_{i+1} = \tilde{w}$ and by
definition of $\prefs{W}$, she does not
blacklist $m_i$. Thus, the algorithm stops, yielding $d=i+1$.
Thus, $w_{i+1}$ and $d$ agree with the definition of
$\abstractcycle{\tilde{w}}$.
Otherwise, $\tilde{w} \ne \matching(m_i)$ , and thus, by definition 
of $\prefs{W}$, $\tilde{w}$ blacklists $w$, and thus $w_{i+1} \ne \tilde{w}$,
and we have $w_{i+1}=\matching(m_i)$. Let $t'$ be the night on which
$m_i$ serenades under $w_{i+1}$'s window for the first time.
By the induction hypothesis and
by \crefpart{abstract-cycle-properties}{simple},
we have that $w_{i+1}$ has yet to reject anyone by night $t$, and thus
by night $t'-1$, and thus on night $t'$ she rejects
the man serenading under her window continuously from the first night, namely
$\matching'(w_{i+1})$, and thus $m_{i+1}=\matching'(w_{i+1})$ and the proof
by induction is complete.
\end{proof}

\begin{cor}\label{build-base-opt}
Under the conditions of \cref{build-base},
$\tilderesultmatching(w)=\matching(w)$ for all women $w$ in
$\abstractcycle{\tilde{w}}$ (and thus in particular also for $w=\tilde{w}$),
as well as for all $w$ for whom $\matching'(w)=\matching(w)$.
\end{cor}

\begin{proof}
By the proof of \cref{build-base}, or immediately by
\crefpart{pm-compatible-properties}{opt-if-reject-cycle} and by
\crefpart{pm-compatible}{stay-or-opt}.
\end{proof}

\begin{lemma}\label{build-step}
Let $\matching'$ and $\matching$ be matchings and let $\prefs{M}$ be a
$(\matching'\rightarrow\matching)$-compatible profile of preference
lists for $M$. Let $\prefs{W}$ be a $(\matching'\xrightarrow{\prefs{M}}\matching)$-compatible profile of preference lists for $W$.
Let $\tilde{w} \in W$ be a woman who rejects a man $\tilde{m}$ during $\run$,
but s.t.\ $\resultmatching(\tilde{w})\ne\matching(\tilde{w})$.
Let $\tildeprefs{W}$ be a profile of preference lists for $W$, obtained
from $\prefs{W}$
by modifying the preference list of $\tilde{w}$ to start with $\matching(\tilde{w})$,
followed immediately by $\tilde{m}$, followed immediately by $\matching'(\tilde{w})$,
followed by some or all other men in arbitrary order. $\tildeprefs{W}$ satisfies both of the following.
\begin{parts}
\item\label{build-step-compatible}
$\tildeprefs{W}$ is $(\matching'\xrightarrow{\prefs{M}}\matching)$-compatible.
\item\label{build-step-cycle}
$\tilderejectcycle=(w_1\xrightarrow{m_1} w_2 \cdots \xrightarrow{m_{\ell-1}} w_{\ell} \xrightarrow{m_{\ell}} \abstractcycle{\tilde{w}} \xrightarrow{m_{\ell}} w_{\ell+1} \xrightarrow{m_{\ell+1}} w_{\ell+2} \xrightarrow{m_{\ell+2}} \cdots w_d)$, where
$\generalcycle\eqdef\rejectcycle$ and where $m_{\ell}=\tilde{m}$ is rejected
by $\tilde{w}$ during $\run$ between his rejection
by $w_{\ell}$ and his provisional acceptance by $w_{\ell+1}$.
\end{parts}
\end{lemma}

\begin{proof}
It is straightforward to check, immediately from the definition of
$\tildeprefs{W}$, that
since
$\prefs{W}$ and $\prefs{M}$ are $\matching'$-cycle generating, so are
$\tildeprefs{W}$ and $\prefs{M}$ (and with the same cycle trigger and
$w_m$'s).
Before continuing to prove \cref{build-step-compatible}, we now prove
\cref{build-step-cycle}.

By \crefpart{pm-compatible-properties}{opt-if-reject-cycle}, $\tilde{w} \notin \{w_i\}_{i=1}^d$; thus, $\tilde{w}$
rejects $\tilde{m}$ during $\run$ ``on his way'' between two women $w_{\ell}$
and $w_{\ell+1}$ for some $\ell$, and thus $m_{\ell}=\tilde{m}$.
As $\tilde{m}$ does not serenade under $\tilde{w}$'s window after she rejects
him, $\ell$ is well-defined.
We note that by definition of $\tildeprefs{W}$, $\tilderun$ and $\run$
coincide as long as $\tilde{m}$ does not does serenade under $\tilde{w}$'s
window, and thus $\tilderejectcycle$ has prefix
$(w_1\xrightarrow{m_1} w_2 \cdots \xrightarrow{m_{\ell-1}} w_{\ell} \xrightarrow{m_{\ell}}$.

By \crefpart{pm-compatible-properties}{disjoint-union},
and by \crefpart{abstract-cycle-properties}{disjoint-or-rotate}, we have that for every $m$ in $\abstractcycle{\tilde{w}}$,
$m$ is not in $\rejectcycle$; thus, such $m$ is never rejected during $\run$,
and thus serenades throughout all nights
of $\run$ under the window of the same woman --- the one under whose
window $m$ serenades on the first night of $\run$, namely $\matching'(m)$,
by definition.

We now show by induction that $\tilderejectcycle$ has prefix
$(w_1\xrightarrow{m_1} w_2 \cdots \xrightarrow{m_{\ell-1}} w_{\ell} \xrightarrow{m_{\ell}} \abstractcycle{\tilde{w}}$.
Denote $\abstractcycle{\tilde{w}}=\generalcycleprime$.

Base:
As explained above, up to the night on which $\tilde{m}$ serenades under
$\tilde{w}$'s window, we have that $\tilderejectcycle$ has prefix
$(w_1\xrightarrow{m_1} w_2 \cdots \xrightarrow{m_{\ell-1}} w_{\ell} \xrightarrow{m_{\ell}}$. As explained above, the other man serenading under $\tilde{w}$'s
window on that night is $\matching'(\tilde{w})$, and by definition of
$\tildeprefs{W}$, she thus
rejects him in favour of $\tilde{m}$ during that night of $\tilderun$. Thus, the next woman
in $\tilderejectcycle$ is $\tilde{w}=w_1'$, and the next man --- $\tilde{m}=m_1'$.

Step:
Let $1 < i \le d'$ and assume that we have shown that $\tilderejectcycle$
has prefix 
$(w_1\xrightarrow{m_1} w_2 \cdots \xrightarrow{m_{\ell-1}} w_{\ell} \xrightarrow{m_{\ell} } w_1' \xrightarrow{m_1'} \cdots w_i' \xrightarrow{m_i'}$.
As explained above, we have $\resultmatching(m_i') \ne \matching(m_i')$,
and thus, by \crefpart{pm-compatible}{trigger} (since
$w_1 \ne \tilde{w}$), we
have that $m_i'$ is blacklisted by $w_1$. Thus, the next woman
in $\tilderejectcycle$ is not $w_1$.
Once more by the same well-known property of the
Gale-Shapley algorithm as above, on every night each woman $w \in W \setminus \{w_1\}$ is serenaded-to by a man
she weakly prefers over $\matching'(w)$.
Thus, by \crefpart{pm-compatible-properties}{only-harbor}, and by \cref{all-matched}, we have that the next woman in $\tilderejectcycle$ is $\matching(m_i')=w_{i+1}'$.
If $i=d$, the proof by induction is complete. Otherwise,
let $t$ be the night on which $m_i'$
first serenades under $w_{i+1}'$'s window during $\tilderun$.
As $m_{i+1}'=\matching'(w_{i+1}')$ is not in $\rejectcycle$ and by
\crefpart{abstract-cycle-properties}{simple},
we have that $m_{i+1}'$ has not been rejected before night $t$,
and thus he is the other man (in addition to $m_i'$) serenading under
$\matching'(w_{i+1})$'s window on night $t$. Therefore, by \crefpart{pm-compatible}{idol} (and since $w_{i+1}' \ne \tilde{w}$,
by \crefpart{abstract-cycle-properties}{simple}), the next man in $\tilderejectcycle$ is $m_{i+1}'$,
and the proof by induction is complete.

Denote by $t$ the night on which $\tilde{m}$ first serenades under $\tilde{w}$'s
window and by $t'$ --- the night on which $m_{d'-1}'$ first serenades under her window.
As $\rejectcycle$ and $\abstractcycle{\tilde{w}}$ are disjoint, and as
$\tilde{m}=m_{\ell}$ is in $\rejectcycle$, it follows that $\tilde{m}$
is not rejected between nights $t$ and $t'$, exclusive of $t'$. Thus,
and as $w_{d'}'=w_1'=\tilde{w}$ provisionally accepts $\tilde{m}$ at time $t$,
it follows that he is the man rejected by $w_{d'}'$ in favour of $m_{d'-1}'$
during $\tilderun$.

So far, we have that $\tilderejectcycle$ has prefix
$(w_1\xrightarrow{m_1} w_2 \cdots \xrightarrow{m_{\ell-1}} w_{\ell} \xrightarrow{m_{\ell}} \abstractcycle{\tilde{w}} \xrightarrow{m_{\ell}}$.
We conclude the proof inductively.

Assume that for some $\ell \le i < d$, $\tilderejectcycle$ has prefix
$(w_1\xrightarrow{m_1} w_2 \cdots \xrightarrow{m_{\ell-1}} w_{\ell} \xrightarrow{m_{\ell}} \abstractcycle{\tilde{w}} \xrightarrow{m_{\ell}}
w_{\ell} \xrightarrow{m_{\ell+1}} \cdots w_i \xrightarrow{m_i}$.

If $i=\ell$, denote $r=\tilde{w}=w_{d'}'$; otherwise, denote $r=w_i$.
Let $t$ be the night on which $r$ rejects $m_i$ during $\run$ and
let $t'$ be the night on which $r$ rejects $m_i$ during $\tilderun$.
In order to show that the next woman in $\tilderejectcycle$ is $w_{i+1}$,
we have to show that on night $t'$ in $\tilderun$, every woman 
that $m_i$ prefers over $w_{i+1}$ but less than $r$, prefers 
her provisional match over $m_i$, while $w_{i+1}$ prefers
$m_i$ over her provisional match.
Let $w \in W$ s.t.\ $m_i$ prefers $r$ over $w$ and weakly prefers $w$
over $w_{i+1}$.
If $w$ is in $\abstractcycle{\tilde{w}}$, then as $\abstractcycle{\tilde{w}}$
and $\rejectcycle$ are distinct, $w \ne w_{i+1}$.
Furthermore, as $w$ is in $\abstractcycle{\tilde{w}}$, then by the induction
hypothesis
she is provisionally matched
on some night before night $t'$ in $\tilderun$ with $\matching(w)$,
and by \crefpart{pm-compatible}{idol} and by definition
of $\tildeprefs{W}$ prefers him over any other man (including $m_i$), and
is thus, in particular, also provisionally matched with $\matching(w)$ on night $t'$.
Otherwise, $w$ is not in $\abstractcycle{\tilde{w}}$ and thus in particular
$w \ne \tilde{w}$.
As $w \ne \tilde{w}$, she has the same preferences according to $\tildeprefs{W}$
as she does according to $\prefs{W}$, and in particular
$m_i$ serenades under $w$'s window ``on his way'' to $w_{i+1}$ (inclusive) during
$\run$, and thus it is enough
to show that she has the same provisional match on night $t$ in $\run$
as she does on night $t'$ in $\tilderun$. Indeed, if there does not exists $j<i$ s.t.\
$w=w_j$, then by the induction hypothesis this provisional match on both nights is
$\matching'(w)$; otherwise, let $j<i$ be maximal s.t.\ $w=w_j$ ---
if $j>1$ then this provisional match on both nights in both runs is $m_{j-1}$
and otherwise she is not provisionally matched to anyone in either night in
either run.
Thus, we conclude that
the next woman in $\tilderejectcycle$ is indeed $w_{i+1}$.

If $w_{i+1}=\tilde{w}$, then both $\run$ and $\tilderun$ stop,
and thus $i+1=d$ and the proof is complete; assume, therefore, that $w_{i+1}\ne
\tilde{w}$.
Recall that $\abstractcycle{\tilde{w}}$ and $\rejectcycle$ are disjoint;
thus, $w_{i+1}$ is not in $\abstractcycle{\tilde{w}}$ and therefore,
as explained above, her provisional match in $\tilderun$ at night $t'$
is the same as in $\run$ at night $t$, namely $m_{i+1}$, and thus he
is the next man in $\tilderejectcycle$ and the proof by
induction is complete.

We now turn to finish proving \cref{build-step-compatible}.
\crefshowpart{pm-compatible}{idol} holds for $\tildeprefs{W}$ since it
holds for $\prefs{W}$ and by definition of $\tildeprefs{W}$.

To show that \cref{pm-compatible-trigger,pm-compatible-non-trigger,pm-compatible-stay-or-opt} of \cref{pm-compatible} hold,
it is enough to show that for every man $m \in M$ s.t.\ $\tilderesultmatching(m) \ne
\resultmatching(m)$, we have $\tilderesultmatching(m) = \matching(m)$.
(To deduce \crefshowpart{pm-compatible}{non-trigger},
we note that $\tilderesultmatching(\tilde{m})=\matching(m)$, as seen above.)
Let $m \in M$. If $m$ is in $\rejectcycle$, then let $i$ be maximal s.t.\ $m_i=m$.
As explained above, $\tilderesultmatching(m)=w_i=\resultmatching(m)$.
(And by \crefpart{pm-compatible-properties}{opt-if-reject-cycle}, $\resultmatching(m)=\matching(m)$ anyway.)
Otherwise, if $m$ is in $\tilderejectcycle$, then $m$ is in $\abstractcycle{\tilde{w}}$,
and as shown above, $\tilderesultmatching(m)=\matching(m)$. Finally,
if $m$ is not in $\tilderejectcycle$, then $m$ is not in $\rejectcycle$
either, and by \cref{changed-if-reject-cycle}, we have
$\tilderesultmatching(m)=\matching'(m)=\resultmatching(m)$ and the proof is complete.
\end{proof}

\begin{cor}\label{build-opt-increase}
Under the conditions of \cref{build-step},
$\tilderesultmatching(w)=\matching(w)$ for all women $w$ in
$\abstractcycle{\tilde{w}}$ (and thus in particular also for $w=\tilde{w}$),
as well as for all $w$ for whom $\resultmatching(w)=\matching(w)$.
\end{cor}

\begin{proof}
By the proof of \cref{build-step}, or immediately by
\crefpart{pm-compatible-properties}{opt-if-reject-cycle} and by
\crefpart{pm-compatible}{stay-or-opt}.
\end{proof}

We are now ready to define the inductive result stemming from
\cref{build-base} as the induction base and from
\cref{build-step} as the induction step.

\begin{lemma}\label{build-cor}
Let $\matching'$ and $\matching$ be matchings,
let $\prefs{M}$ be a $(\matching'\rightarrow\matching)$-compatible
profile of preference lists for $M$ and let
$\prefs{W}$ be a profile
of preference lists for $W$ according to which each $w \in W$ has preference list
starting with $\matching(w)$, followed immediately by $\matching'(w)$
(if $\matching'(w)\ne\matching(w)$), followed by some or all
other men in arbitrary order.
For every $\tilde{w} \in W$
s.t.\ $\matching(\tilde{w}) \ne \matching'(\tilde{w})$,
there exists a profile $\tildeprefs{W}$ of preference lists for $W$
s.t.\ all of the following hold.
\begin{parts}
\item\label{build-cor-opt}
$\tilderesultmatching(\tilde{w})=\matching(\tilde{w})$, and $\tilde{w}$
still has $\matching(\tilde{w})$ at the top of her preference list, followed
only by men who do not serenade under $\tilde{w}$'s window during $\tilderun$.
\item\label{build-cor-diff-then-opt}
Every $w \in W \setminus \{\tilde{w}\}$ whose preference lists according to $\prefs{W}$ and 
to $\tildeprefs{W}$ differ satisfies all of the following.
\begin{parts}
\item
$w$ still has $\matching(w)$ at the top of her preference list.
\item
The difference in $w$'s preference list is only in the promotion of a man
$m$ who strictly prefer $\matching'(m)$ over $w$.
\item
$\tilderesultmatching(w)=\matching(w)\ne\matching'(w)$
\end{parts}
\item\label{build-cor-opt-then-opt}
$\bigl\{w \in W \mid \matching'(w)=\matching(w)\bigr\}
\subsetneq
\bigl\{w \in W \mid \tilderesultmatching(w)=\matching(w)\bigr\}$.
\item\label{build-cor-reject-then-opt}
Every $w \in W$ who rejects any man during $\tilderun$
satisfies $\tilderesultmatching(w)=\matching(w)$.
\item\label{build-cor-blacklist-then-opt}
No woman's blacklist differs between $\prefs{W}$ and $\tildeprefs{W}$, except
perhaps for that of $\tilde{w}$, each of the men $m$ she blacklists
according to $\tildeprefs{W}$
satisfying $\tilderesultmatching(m)=\matching(m)\ne\matching'(m)$.
\end{parts}
\end{lemma}

\begin{proof}
Let $\tildeprefs{W}$ be the profile of preferences obtained by
$\prefs{W}$ by modifying the preference list of $\tilde{w}$
to consist solely of $\matching(\tilde{w})$, with all other men blacklisted.
By \cref{build-base}, $\tildeprefs{W}$ is
$(\matching'\xrightarrow{\prefs{M}}\matching)$-compatible.
As long as there exits a woman $w \in W$ who rejects
a man during $\tilderun$ but s.t.\ $\tilderesultmatching(w) \ne \matching(w)$,
we may modify $\tildeprefs{W}$ using \cref{build-step} (without modifying
$w$'s blacklist), keeping it $(\matching'\xrightarrow{\prefs{M}}\matching)$-compatible and only enlarging $\tilderejectcycle$. We thus repeatedly
apply \cref{build-step} to modify $\tildeprefs{W}$ until no such woman exists.
We note that this process stops after at most $\left\lfloor\frac{n-1}{2}\right\rfloor$ applications
of \cref{build-step},
as by \cref{build-opt-increase} the number
of women $w \in W$ for whom $\tilderesultmatching(w)=\matching(w)$ increases
by at least $2$ following each application of \cref{build-step} (and 
by \cref{build-base-opt} equals at
least $2$ immediately after applying \cref{build-base}).
Finally, we shorten the blacklist of $\tilde{w}$ in $\tildeprefs{W}$ to contain
only the men that she actually rejects during $\tilderun$
($\tilderun$ is thus completely unaffected
by this).

By \cref{build-base-opt,build-opt-increase}, we have that
that \cref{build-cor-opt,build-cor-opt-then-opt} hold. (The second part
of \cref{build-cor-opt} follows from the final stage of the construction
of $\tildeprefs{W}$.)
We now show that \cref{build-cor-diff-then-opt} holds. 
Let $w \in W \setminus \{\tilde{w}\}$  whose preference list differs between
$\prefs{W}$ and $\tildeprefs{W}$. By our construction, $w$'s preference list
was changed by an application of
\cref{build-step} with $w$ in the role of $\tilde{w}$ (as defined there),
and therefore $\matching'(w)\ne\matching(w)$, and
by \cref{build-opt-increase},
also $\tilderesultmatching(w)=\matching(w)$. Furthermore, the change of
$w$'s preference list is in the promotion to second place
of a man $\tilde{m}$ who,
by \crefpart{build-step}{cycle},
serenades under $w$'s window during $\tilderun$ following his
rejection by another woman; thus, $\tilde{m}$ strictly prefers
$\matching'(m)$, under whose window he serenades on the first night
of $\tilderun$, over $w$, as required.

\cref{build-cor-reject-then-opt} follows by construction, as
it was no longer possible to apply \cref{build-step} when the
construction stopped.

It thus remains to prove \cref{build-cor-blacklist-then-opt}.
By construction, the only blacklist that was changed during construction
of $\tildeprefs{W}$ from $\prefs{W}$ is indeed that of $\tilde{w}$.
Furthermore, by the last step of the construction, every man blacklisted
by $\tilde{w}$ is rejected by her during $\tilderun$. Thus, every such
man is in $\tilderejectcycle$ and
by \crefpart{pm-compatible-properties}{opt-if-reject-cycle} the proof is complete.
\end{proof}

\begin{cor}
Under the conditions of \cref{build-cor}, denote the size of $\tilde{w}$'s
blacklist, according to $\tildeprefs{W}$, by $b$. At least $b+1$ men $m$
satisfy $\tilderesultmatching(m)=\matching(m)\ne\matching'(m)$.
\end{cor}

\begin{proof}
By \crefpart{build-cor}{blacklist-then-opt},
all $b$ distinct men in $\tilde{w}$'s blacklist (according to $\tildeprefs{W}$)
satisfy the required condition.
Furthermore, by definition of $\tilde{w}$ and by
\crefpart{build-cor}{opt}, we have that
$\tilderesultmatching(\tilde{w})=\matching(\tilde{w})\ne\matching'(\tilde{w})$, and thus
$\tilderesultmatching(\tilde{w})$, who is by definition not blacklisted by
$\tilde{w}$ and is thus distinct from the $b$ men in her blacklist,
satisfies the required condition as well.
\end{proof}

While a considerable amount of reasoning went into proving
\cref{build-cor} (and the lemmas supporting it), and while $\prefs{W}$ is
constructed in its proof via an inductive process
na\"{\i}vely requiring a simulation of a run of the Gale-Shapley algorithm on
every iteration, we now show that $\prefs{W}$ can be calculated quite efficiently,
by leveraging the theory of cycles developed above. Indeed, without the aid
of this theory, \cref{nk} may seem quite obscure.

\begin{lemma}\label{nk-proposition}
$\tildeprefs{W}$ from \cref{build-cor} can be computed in $O(n \cdot k)$ time,
where $k$ is the number of men $m$ (alternatively, women) satisfying
 $\tilderesultmatching(m)=\matching(m)\ne\matching'(m)$.
\end{lemma}

\newcommand{\nkinitproc}{Initialize}
\newcommand{\nkproc}{Compute\namecref{build-cor}\labelcref{build-cor}}

\algnewcommand{\variable}[1]{\ensuremath{\mathit{#1}}}
\algnewcommand{\todo}{\variable{todo}}
\algnewcommand{\shouldBlacklist}{\variable{shouldBlacklist}}
\algnewcommand{\upcoming}{\variable{currentPref}}
\algnewcommand{\provisional}{\variable{provisionalMatching}}
\algnewcommand{\swap}{\variable{provisionallyAccept}}
\algnewcommand{\false}{{\bf false}}
\algnewcommand{\true}{{\bf true}}

\begin{proof}
\cref{nk} computes $\tildeprefs{W}$, as it is equivalent
to the construction of \cref{build-cor}, each time applying
\cref{build-step} to the woman who performs the earliest rejection out
of the women $w$ for who are not eventually-matched with $\matching(w)$.

\begin{center}
\captionof{algorithm}{In-place computation of $\tildeprefs{W}$ from \cref{build-cor}, using $O(n)$ additional memory, in $O(n \cdot k)$ time. (All lines run in
$O(1)$ time, unless otherwise noted.)}
\begin{algorithmic}[1]\label{nk}

\LineComment{See \cref{nn-flat-case} for the reason for separation of the initialization.}
\Procedure{\nkinitproc}{$\matching',\matching$}\label{nk-init}
\LineComment{Initialize \todo.}
\State $\todo \gets \bigl\{w \in W \mid \matching'(w)\ne\matching(w)\bigr\}$\Comment{implemented as hash set; $O(n)$}
\EndProcedure

\Procedure{\nkproc}{$\matching',\matching,\prefs{M},\tilde{w},\tildeprefs{W}$}\label{nk-main}\Comment{$\tildeprefs{W}=\prefs{W}$ when called}
\LineComment{Initialize $\tilderun$ simulation.}
\State $\provisional \gets \matching'$\Comment{$O(n)$}
\State $\shouldBlacklist \gets$ empty hash set\Comment{the men to be blacklisted by $\tilde{w}$}
\State $\upcoming\gets$ array indexed by $M$\Comment{cache, initialized in lines \ref{nk-cache-init-1} and \ref{nk-cache-init-2}}

\LineComment{Initialize $\tildeprefs{W}$ as in \cref{build-base}.}
\State set preference list of $\tilde{w}$ in $\tildeprefs{W}$ to $\bigl(\matching(\tilde{w})\bigr)$\Comment{not required for correctness}

\LineComment{Simulate $\tilderun$, adjusting $\tildeprefs{W}$ as needed.}
\State insert $m$ into $\shouldBlacklist$
\State $\provisional[\tilde{w}] \gets \emptyset$\Comment{instead of $m$; not required for correctness}
\State $m \gets \matching'(\tilde{w})$\Comment{$m$ is the man currently ``on the move''}
\State \Call{MarkAsDoneCycleOf}{$\tilde{w}$}\Comment{see line~\ref{nk-mark-procedure}}
\State\label{nk-cache-init-1} $\upcoming[m] \gets$ pointer to $m$'s preference list right after $\tilde{w}$\Comment{$O(n)$}
\While{$m \ne \matching(\tilde{w})$ {\bf or} deref $\upcoming[m] \ne \tilde{w}$}\label{nk-main-while}\Comment{$< k \cdot (n-1)$ iterations}
\LineComment{Loop Invariant:}
\LineComment{$\todo = \bigl\{w \in W \mid \tilderesultmatching(w)\ne\matching(w)\bigr\}$.}
\State $w \gets$ deref $\upcoming[m]$\Comment{the woman currently approached by $m$}
\State advance $\upcoming[m]$
\State let $\swap$ be a boolean\Comment{should $w$ provisionally accept $m$?}
\If{$w=\tilde{w}$}\Comment{by the loop condition, $m \ne \matching(\tilde{w})$}
\State insert $m$ into $\shouldBlacklist$
\State $\swap \gets \false$
\ElsIf{$m = \matching(w)$}
\State $\swap \gets \true$
\State discard $\upcoming[m]$
\ElsIf{$w \in \todo$}\Comment{$\le \left\lfloor\frac{k-1}{2}\right\rfloor$ times throughout the run}
\LineComment{Update $\tildeprefs{W}$ as in \cref{build-step}.}
\State promote $m$ to be second on $w$'s preference list in $\tildeprefs{W}$\Comment{$O(n)$}
\State \Call{MarkAsDoneCycleOf}{$w$}\Comment{see line~\ref{nk-mark-procedure}}
\State $\swap \gets \true$
\Else
\State $\swap \gets \false$
\EndIf
\If{$\swap$}
\State swap values of $m$ and $\provisional[w]$
\If{undefined $\upcoming[m]$}\Comment{$< k$ times throughout the run}
\State\label{nk-cache-init-2} $\upcoming[m] \gets$ pointer to $m$'s preference list right after $w$\Comment{$O(n)$}
\EndIf
\EndIf\Comment{else, do not provisionally accept (i.e. $m$ is rejected by $w$)}
\EndWhile
\State $\provisional[\tilde{w}] \gets m$
\LineComment{Update $\tilde{w}$'s blacklist per the final step of the construction from \cref{build-cor}.}
\State set $\tilde{w}$'s blacklist to be $\shouldBlacklist$ (keeping $\matching(\tilde{w})$ most-preferred)\Comment{$O(n)$}
\LineComment{At this point, $\provisional=\tilderesultmatching$.}
\State \Return{\provisional}
\EndProcedure

\LineComment{Removes all women in $\abstractcycle{w}$ from \todo.}
\Procedure{MarkAsDoneCycleOf}{$w$}\label{nk-mark-procedure}
\State $w' \gets w$
\Repeat\Comment{$\le n$ iterations throughout the run of \Call{\nkproc}{}}
\State remove $w'$ from \todo
\State $w' \gets \matching(\matching'(w'))$;
\Until{$w'=w$}
\EndProcedure

\end{algorithmic}
\end{center}
We note that as there are at most $k$ rejected men (by
\crefpart{pm-compatible-properties}{opt-if-reject-cycle}), each newly-serenading at most $n-1$ times,
the loop starting at line~\ref{nk-main-while}
indeed iterates less than $k \cdot (n-1)$ times.
\end{proof}

Before proving \cref{inconspicuous}, we first prove a simpler special
case thereof, in which each man approaches a distinct woman on the first night,
e.g.\ as in \cref{men-forcing}.
(We reference and reuse significant portions of this proof in the more-intricate
proof 
of the general case of \cref{inconspicuous} below.)

\begin{thm}\label{flat-case}
\cref{inconspicuous} holds for the special case in which
each man approaches a distinct woman on the first night.
(In this case, in the absence of nonempty blacklists, the algorithm would
terminate at the end of the first night.) Furthermore, in this case $\prefs{W}$
can be computed online in $O(n^2)$ time.
\end{thm}

\newcommand{\nnflatcaseproc}{Compute\namecref{inconspicuous}\labelcref{inconspicuous}SpecialCase}

\begin{proof}
Let $\prefs{M}$ and $\matching$ be as in
\cref{inconspicuous}, but s.t.\ the top choices of all men are distinct.
We inductively define a sequence of profiles $(\prefs{W}^i)_{i=1}^d$
of preference lists for $W$
and a sequence of matchings $(\matching^i)_{i=1}^d$, for $d$ to be defined below,
satisfying the following properties for every
$i\ge1$.
\begin{properties}
\item\label{flat-case-matching}
$\matching^i=\menopt(\prefs{W}^i,\prefs{M})$.
\item\label{flat-case-increase}
If $i>1$, then
$\bigl\{w \in W \mid \matching^{i-1}(w)=\matching(w)\bigr\}
\subsetneq
\bigl\{w \in W \mid \matching^i(w)=\matching(w)\bigr\}$.
\item\label{flat-case-compatible}
$\prefs{M}$ is $(\matching^i\rightarrow\matching)$-compatible.
\item\label{flat-case-idol}
Each $w \in W$ has $\matching(w)$ first on her preference list according to $\prefs{W}^i$.
\item\label{flat-case-second-rank}
Each $w \in W$ for whom $\matching^i(w)\ne\matching(w)$ has
$\matching^i(w)$ second on her preference list (after $\matching(w)$)
according to $\prefs{W}^i$.
\item\label{flat-case-reject-then-opt}
Each $w \in W$ for whom $\matching^i(w)\ne\matching(w)$
does not reject any man during $\ifullrun$.
\item\label{flat-case-no-double-blacklist}
No man appears in more than one blacklist, and every blacklisted man $m\in M$
satisfies $\matching^i(m)=\matching(m)$.
\item\label{flat-case-no-blacklist-opt}
For each $w \in W$ who has a nonempty blacklist, $\matching(w)$ appears in no blacklist
and satisfies $\matching(\matching(w))=\matching^i(\matching(w))$.
\end{properties}

Base:
Denote by $\matching^1$ the (provisional) matching describing the first night of the
algorithm. As $\matching$ is $M$-rational, and as each
$m \in M$ weakly prefers $\matching^1(m)$ over any $w \in W$, $\prefs{M}$ is
$(\matching^1\rightarrow\matching)$-compatible.
Denote by $\prefs{W}^1$ a profile of preference lists for $W$ according
to which each $w \in W$ has preference list starting with $\matching(w)$,
followed immediately by $\matching^1(w)$ (if $\matching^1(w)\ne\matching(w)$),
followed by all other men in arbitrary order.
We note that as all blacklists in $\prefs{W}$ are empty, no rejections occur in $\ofullrun$.
In particular, we thus have $\menopt(\prefs{W}^1,\prefs{M})=\matching^1$. \cref{flat-case-no-double-blacklist,flat-case-no-blacklist-opt} follow immediately as all the blacklists in $\prefs{W}^1$ are empty.

Step: Assume that $\matching^i$ and $\prefs{W}^i$ have been defined for some $i\ge 0$.
If $\matching^i=\matching$, we conclude the inductive process and
set $d\eqdef i$ and $\prefs{W}\eqdef\prefs{W}^d$.
Otherwise, let $\tilde{w}^{i+1}$ be a woman s.t.\ $\matching^i(\tilde{w}^{i+1}) \ne \matching(\tilde{w}^{i+1})$.
We denote by $\prefs{W}^{i+1}$ the profile of preference lists for $W$
obtained by use of \cref{build-cor} when applied to
$\matching$, $\prefs{W}^i$, $\prefs{M}$, $\matching'\eqdef\matching^i$ and $\tilde{w}\eqdef\tilde{w}^{i+1}$. Furthermore, we set $\matching^{i+1}\eqdef\iincresultmatching$.
\cref{build-cor} is indeed applicable due to
\cref{flat-case-idol,flat-case-compatible,flat-case-second-rank} for $i$.
We now show that \cref{flat-case-matching,flat-case-increase,flat-case-compatible,flat-case-idol,flat-case-second-rank,flat-case-reject-then-opt,flat-case-no-double-blacklist,flat-case-no-blacklist-opt} hold for $i+1$.

Consider the following timing for $\incfullrun$: In the first part of the run,
denoted by $\firstip$,
the algorithm runs normally, except that the rejection of
$\matching^i(\tilde{w}^{i+1})$ by $\tilde{w}^{i+1}$
does not (yet) take place.
In the second part of the run, denoted by $\secondip$,
after $\firstip$ converges,
$\tilde{w}^{i+1}$'s
rejection of $\matching^i(\tilde{w}^{i+1})$ takes place and the algorithm runs
until it converges once again. As the outcome of the Gale-Shapley algorithm is
invariant under timing changes~\citep{Dubins-Freedman}, this timing also
yields $\menopt(\prefs{W}^{i+1},\prefs{M})$. We claim that $\firstip$
is indistinguishable from $\ifullrun$, while $\secondip$ is indistinguishable
from $\iincrun$.
By \cref{flat-case-reject-then-opt} for $i$, $\tilde{w}^{i+1}$ does not reject
any man during $\ifullrun$, and thus any change to her preference list has
no effect on $\ifullrun$, as long as any blacklist-induced rejection of
$\matching^i(\tilde{w}^{i+1})$ (the only man serenading under $\tilde{w}^{i+1}$'s window during
$\ifullrun$) is deferred.
By \crefpart{build-cor}{diff-then-opt},
the only difference in the preference list of
any other woman $w\ne\tilde{w}^{i+1}$ between
$\prefs{W}^i$ (which yields $\ifullrun$) and $\prefs{W}^{i+1}$ (which
yields $\firstip$), is in the possible promotion of men $m$ who prefer
$\matching^i(m)$ over $w$. As by \cref{flat-case-matching} for $i$,
such men never serenade under $w$'s window
during $\ifullrun$, such a change has no effect on $\ifullrun$ either.
We thus indeed have that $\firstip$
is indistinguishable from $\ifullrun$,
and hence concludes with the provisional matching
$\menopt(\prefs{W}^i,\prefs{M})=\matching^i$; thus,
$\secondip$ is indistinguishable from
$\iincrun$, by definition of the latter. Therefore,
$\menopt(\prefs{W}^{i+1},\prefs{M})=\iincresultmatching$, and
\cref{flat-case-matching} holds for $i+1$.

\cref{flat-case-increase} for $i+1$ follows directly from
\crefpart{build-cor}{opt-then-opt}.
By \cref{flat-case-idol} for $i$, and by
\cref{build-cor}(\labelcref{build-cor-opt,build-cor-diff-then-opt}),
we have that \cref{flat-case-idol} holds for $i+1$
as well.
\cref{flat-case-compatible} for $i+1$ follows from \cref{flat-case-compatible}
for $i$ and from \cref{flat-case-idol} for $i+1$, by \cref{still-compatible}
and as $\matching^{i+1}=\iincresultmatching$.

Let $w \in W$ s.t.\ $\matching^{i+1}(w) \ne \matching(w)$; 
thus, by \crefpart{build-cor}{opt}, $w \ne \tilde{w}^{i+1}$. By
\cref{flat-case-increase} for $i+1$, $\matching^{i}(w) \ne \matching(w)$ as well.
Thus, by \crefpart{build-cor}{diff-then-opt}, we have
that \cref{flat-case-second-rank} for $i+1$ follows
from \cref{flat-case-second-rank} for $i$.
\cref{flat-case-reject-then-opt} for $i+1$ follows from
\cref{flat-case-reject-then-opt} for $i$ (for $\firstip$),
and from \crefpart{build-cor}{reject-then-opt} (for $\secondip$).

To prove \cref{flat-case-no-double-blacklist} for $i+1$ from
\cref{flat-case-no-double-blacklist} for $i$, we must show that every
newly-blacklisted man $m$ is newly-blacklisted by exactly one woman,
is not blacklisted in $\prefs{W}^i$, and
satisfies $\matching^{i+1}(m)=\matching(m)$. Indeed,
by \crefpart{build-cor}{blacklist-then-opt}, any newly-blacklisted man $m$
is newly-blacklisted only by $\tilde{w}^{i+1}$ and satisfies $\matching^{i+1}(m)=\matching(m) \ne \matching^i(m)$. Thus, in particular, by \cref{flat-case-no-double-blacklist} for $i$, any such $m$ is not blacklisted in $\prefs{W}^i$.

Finally, since by \crefpart{build-cor}{blacklist-then-opt} no
blacklist is changed but $\tilde{w}^{i+1}$'s,
in order to prove \cref{flat-case-no-blacklist-opt} for $i+1$ given
\cref{flat-case-no-blacklist-opt} for $i$,
it is enough to show that no woman
blacklists $\matching(\tilde{w}^{i+1})$ in $\prefs{W}^i$, and that there
exists no woman $w$ with a nonempty blacklist in $\prefs{W}^i$ s.t.\
$\matching(w)$ is blacklisted by $\tilde{w}^{i+1}$ in $\prefs{W}^{i+1}$.
The former holds since
by definition of $\tilde{w}^{i+1}$ we have $\matching^i(\matching(\tilde{w}^{i+1}))
\ne \tilde{w}^{i+1}$, and thus by \cref{flat-case-no-double-blacklist} for $i$,
$\matching(\tilde{w}^{i+1})$ is not blacklisted in $\prefs{W}^i$. To show 
the latter, let $w$ be a woman with a nonempty blacklist in $\prefs{W}^i$. By
\cref{flat-case-no-blacklist-opt} for $i$, $\matching^i(\matching(w))=\matching(\matching(w))$,
and thus, by \crefpart{build-cor}{blacklist-then-opt},
$\matching(w)$ is not blacklisted by $\tilde{w}^{i+1}$ in $\prefs{W}^{i+1}$
and the proof of \cref{flat-case-no-blacklist-opt} for $i+1$ is complete. Thus,
the proof by induction is complete, as the process stops by \cref{flat-case-increase}, by finiteness of $W$.

We conclude the proof by showing that $\prefs{W}=\prefs{W}^d$
satisfies the conditions of \cref{inconspicuous}.
By definition of $d$,
$\matching^d=\matching$; thus, by \cref{flat-case-matching} for $i=d$,
$\menopt(\prefs{W},\prefs{M})=\matching$. Furthermore, the $W$-optimal
stable matching is also $\matching$, by \cref{flat-case-idol} for $i=d$.
As the $W$-optimal and $M$-optimal stable
matchings coincide, and as the latter is also the $W$-worst stable matching, these constitute the unique stable matching.

By \cref{flat-case-no-double-blacklist} for $i=d$,
each man is blacklisted by at most one woman in $\prefs{W}$.
Let $w_1,\ldots,w_{n_b}$ be the women with nonempty blacklists in $\prefs{W}$.
By \cref{flat-case-no-blacklist-opt} for $i=d$,
$m_1=\matching(w_1),\ldots,m_{n_b}=\matching(w_{n_b})$
constitute $n_b$ distinct men not blacklisted in $\prefs{W}$. Hence, there exist
at most $|M|-n_b=n-n_b$ men blacklisted in $\prefs{W}$. Thus, and since
every man is blacklisted by at most one woman, we have that the combined
size of all blacklists is at most $n-n_b$.
As each of $w_1,\ldots,w_{n_b}$ has a blacklist of size at least 1,
we have that the combined size of all blacklists is at least $n_b$.
Combining these, we obtain $n_b \le n-n_b$,
yielding $n_b \le \left\lfloor \frac{n}{2} \right\rfloor$ and completing the proof.

\cref{nn-flat-case} follows the above construction for computing $\prefs{W}$.
This algorithm simulates the run $\fullrun$ \citep[or more precisely, an equivalent run of the implementation
of the Gale-Shapley algorithm proposed by][]{Dubins-Freedman}, while building $\prefs{W}$ online;
i.e.\ the choice of who acts next and the decisions of whom each woman prefers or blacklists at any step
depend solely on the history of the run
(and not on the not-yet-acted-upon suffixes of the men's preference lists).
Thus, if participants are not required to submit their preference lists
in advance, but rather only to dynamically act upon them, then in \citeauthor{Dubins-Freedman}'s implementation
of the algorithm, if the women can control its scheduling, \cref{nn-flat-case} constitutes
a strategy for the women
that forces $\matching$ against \emph{every profile}
$\prefs{M}$ of preference lists for~$M$.

\begin{center}
\captionof{algorithm}{On-line computation of $\prefs{W}$ from \cref{inconspicuous}, for the special case of \cref{flat-case}, using $O(n)$ additional memory, in $O(n^2)$ time. (All lines run in $O(1)$ time, unless otherwise noted.)}
\begin{algorithmic}[1]\label{nn-flat-case}

\Procedure{\nnflatcaseproc}{$\prefs{M},\matching$}
\LineComment{Initialize $\prefs{W}$ to $\prefs{W}^1$ and $\matching'$ to $\matching^1$.}
\State $\matching' \gets$ matching between $W$ and $M$\Comment{implemented e.g.\ using two arrays}
\For{$m \in M$}\Comment{$n$ iterations}
\State match $m$ in $\matching'$ with its most-preferred woman according to $\prefs{M}$
\EndFor
\State $\prefs{W}\gets$ profile of preference lists for $W$
\For{$w \in W$}\Comment{$n$ iterations}
\State set preference list of $w$ in $\prefs{W}$ to $(\matching(w),\matching'(w),\mbox{all other men})$\Comment{$O(n)$}
\EndFor
\LineComment{Initialize \todo.}
\State \Call{\nkinitproc}{$\matching',\matching$}\Comment{see \cref{nk}; $O(n)$}
\LineComment{Iterate until for every women $w$ we have $\matching(w)=\matching'(w)$.}
\While{$\todo \ne \emptyset$}
\State let $\tilde{w} \in \todo$
\LineComment{Update $\prefs{W}$ from $\prefs{W}^i$ to $\prefs{W}^{i+1}$ and $\matching'$ from $\matching^i$ to $\matching^{i+1}$.}
\State $\matching' \gets$ \Call{\nkproc}{$\matching',\matching,\prefs{M},\tilde{w},\prefs{W}$}\Comment{see \cref{nk}}
\EndWhile
\State \Return $\prefs{W}$
\EndProcedure

\end{algorithmic}
\end{center}

By \cref{nk-proposition}, each call to \Call{\nkproc}{} takes $O(n\cdot k)$ time, where $k$ is the number of women $w$ for whom $\matching'(w)$
is ``fixed'' from $\matching^1w)$ to $\matching(w)$ by this call. Thus, all calls to \Call{\nkproc}{} take $O(n^2)$ time in total.
\end{proof}

\newcommand{\nnnonflatproc}{Compute\namecref{inconspicuous}\labelcref{inconspicuous}}

The basic ideas of the proof of \cref{flat-case} are used in the proof
that we now give for \cref{inconspicuous} as well,
however this proof is somewhat more delicate, as we must choose
$\tilde{w}^{i+1}$ carefully, and more intricately
analyse the inductive step in certain cases.

\begin{proof}[Proof of \cref{inconspicuous}]
Let $\prefs{M}$ and $\matching$ be as in \cref{inconspicuous}.
We once again inductively define a sequence of profiles $(\prefs{W}^i)_{i=1}^d$ of preference lists for $W$ and a sequence of matchings $(\matching^i)_{i=1}^d$, for $d$ to be defined below, satisfying the following, slightly
modified, properties for every $i\ge1$. (Note the omission of \cref{flat-case-reject-then-opt}, which causes us significant hardship below,  and the addition of \cref{non-flat-single-competitor}.)

\begin{properties}
\item\label{non-flat-matching}
$\matching^i=\menopt(\prefs{W}^i,\prefs{M})$.
\item\label{non-flat-increase}
If $i>1$, then
$\bigl\{w \in W \mid \matching^{i-1}(w)=\matching(w)\bigr\}
\subsetneq
\bigl\{w \in W \mid \matching^i(w)=\matching(w)\bigr\}$.
\item\label{non-flat-compatible}
$\prefs{M}$ is $(\matching^i\rightarrow\matching)$-compatible.
\item\label{non-flat-idol}
Each $w \in W$ has $\matching(w)$ first on her preference list according to $\prefs{W}^i$.
\item\label{non-flat-second-rank}
Each $w \in W$ for whom $\matching^i(w)\ne\matching(w)$ has
$\matching^i(w)$ second on her preference list (after $\matching(w)$)
according to $\prefs{W}^i$.
\stepcounter{propertiesi} % Skip the equivalent of flat-case-reject-then-opt.
\item\label{non-flat-no-double-blacklist}
No man appears in more than one blacklist, and every blacklisted man $m\in M$
satisfies $\matching^i(m)=\matching(m)$.
\item\label{non-flat-no-blacklist-opt}
For each $w \in W$ who has a nonempty blacklist, $\matching(w)$ appears in no blacklist
and satisfies $\matching(\matching(w))=\matching^i(\matching(w))$.
\item\label{non-flat-single-competitor}
For each $w \in W$ who is serenaded-to during $\ifullrun$ solely by
$\matching(w)$, $\matching(w)$ appears in no blacklist.
\end{properties}

Base:
Denote by $\prefs{W}^0$ a profile of preference lists for $W$ according
to which each $w \in W$ has preference list starting with $\matching(w)$,
followed by all other men in arbitrary order.
Define $\matching^1=\menopt(\prefs{W}^0,\prefs{M})$.
We note that all men are matched in $\matching^1$. Indeed, let $m \in M$;
similarly to the proof of \cref{all-matched},
since by definition $m$ does not blacklist $\matching(w)$, and since $\matching(w)'s$
top choice is $m$, then $m$ must be matched by $\matching^1$. Furthermore,
we thus obtain that $m$ weakly prefers $\matching^1(m)$ over $\matching(m)$,
and as $\matching$ is by definition $M$-rational,
we have that $\prefs{M}$ is $(\matching^1\rightarrow\matching)$-compatible.
Denote by $\prefs{W}^1$ the profile of preference lists for $W$
obtained from $\prefs{W}^0$ by promoting, for each woman $w$ s.t.\
$\matching^1(w)\ne\matching(w)$, $\matching^1(w)$ to be second on $w$'s
preference list (immediately following $\matching(w)$). Since by definition
each such woman $w$ never rejects $\matching^1(w)$ during
$\zfullrun$, we have that $\zfullrun$ and $\ofullrun$ are indistinguishable.
Thus, $\matching^1=\menopt(\prefs{W}^1,\prefs{M})$.
Similarly to the induction base in the proof of \cref{flat-case},
\cref{non-flat-no-double-blacklist,non-flat-no-blacklist-opt,non-flat-single-competitor} follow immediately as all blacklists exist in $\prefs{W}^1$ are empty.
We note that unlike the scenario analysed in the proof of \cref{flat-case},
we generally have that
\cref{flat-case-reject-then-opt} from that proof
simply does not hold here, even during $\zfullrun$.

Step: Assume that $\matching^i$ and $\prefs{W}^i$ have been defined for some $i \ge 0$.
If $\matching^i=\matching$, then we conclude the inductive process and
set $d\eqdef i$ and $\prefs{W}\eqdef\prefs{W}^d$. Otherwise,
$T^i\eqdef\bigl\{w \in W \mid \matching^i(w) \ne \matching(w)\bigr\}\ne\emptyset$.
For every woman $w$, denote by $t^i(w)$ the last night of $\ifullrun$
on which any man
newly-serenades under $w$'s window. Let $\tilde{w}^{i+1} \in T^i$
with largest $t^i(\tilde{w}^{i+1})$.
(If more that one
such woman exists, choose one of them arbitrarily.)
We denote by $\prefs{W}'$ the profile of preference lists for $W$
obtained by use of \cref{build-cor}\footnote{We note that this proof,
unlike that of \cref{flat-case}, uses \crefpart{build-cor}{reject-then-opt}
solely to show that for every $w \in W$,
$\matching^{i+1}(w) \in \{\matching^i(w),\matching(w)\}$. Thus, the
use of \cref{build-cor} could be replaced by the use of a simple variant
of \cref{build-base}, rendering \cref{build-step,build-cor} unneeded.
Nonetheless, we present these lemmas for several reasons.
First, since they provide insight
into the structure underlying our solution, and
also allow for both a
much cleaner and more structured proof of \cref{flat-case} and a
more gradual introduction of the methods used in the second case of the current
proof; second, since they yield an online algorithm for the case
of \cref{flat-case}, and since
using  \cref{build-cor} in the current proof provides for improved running-time
of the resulting algorithm for many inputs, by
turning iterations of the second case of this proof into (faster)
applications of \cref{build-step} as iterations in the proof of
\cref{build-cor}.} when applied to
$\matching$, $\prefs{W}^i$, $\prefs{M}$,
$\matching'\eqdef\matching^i$ and $\tilde{w}\eqdef\tilde{w}^{i+1}$.
As in the proof of \cref{flat-case}, \cref{build-cor} is applicable due to
\cref{flat-case-idol,flat-case-compatible,flat-case-second-rank} for $i$.
We define $\matching^{i+1}\eqdef\ipresultmatching$, and
consider two separate cases when defining $\prefs{W}^{i+1}$:
\begin{itemize}
\item
If on night $t^i(\tilde{w}^{i+1})$ of $\ifullrun$ only one man serenades under $\tilde{w}^{i+1}$'s
window, then this night is the first on which any man serenades under
her window during $\ifullrun$; thus, $\tilde{w}^{i+1}$ rejects no man before or on that night,
and thus,
by definition of $\tilde{w}^{i+1}$, she rejects no man during $\ifullrun$.\footnote{An implementation note:
in fact, any woman who never rejects any man during $\ifullrun$ may be chosen
as $\tilde{w}^{i+1}$ and treated as detailed in this case, regardless of whether she
has the largest $t^i$ value or not (this condition is only needed in the
second case, in which $\tilde{w}^{i+1}$ rejects men during $\ifullrun$), resulting in an asymptotically-faster algorithm for many inputs.
The reason we nonetheless describe
our proof by choosing $\tilde{w}^{i+1}$ in the
same manner on both cases above is for the sake of simplicity, in order to better
reflect the analogies between these two cases.
}
In this case, as in the inductive step of \cref{flat-case}, we set
$\prefs{W}^{i+1}\eqdef\prefs{W}'$ (and thus
$\matching^{i+1}=\iincresultmatching$).
The proofs of \cref{non-flat-matching,non-flat-increase,non-flat-second-rank,non-flat-idol,non-flat-compatible,non-flat-no-double-blacklist,non-flat-no-blacklist-opt} for $i+1$ follow exactly as in
the inductive step in the proof of \cref{flat-case} ---
we note that these use only \cref{non-flat-matching,non-flat-increase,non-flat-second-rank,non-flat-idol,non-flat-compatible,non-flat-no-double-blacklist,non-flat-no-blacklist-opt} for $i$, and additionally
\cref{flat-case-reject-then-opt} for $i$, but only to show that
$\tilde{w}^{i+1}$ rejects
no man during $\ifullrun$, which we now have by definition of this case,
as explained above.
It thus remains to show that
\cref{non-flat-single-competitor} holds for $i+1$.
By equivalence of $\firstip$ and $\ifullrun$, any woman $w$ serenaded-to during
$\incfullrun$ solely by $\matching(w)$ is also serenaded-to during $\ifullrun$
solely by him. Thus,
\cref{non-flat-single-competitor} for $i+1$ follows from
\cref{non-flat-single-competitor} for $i$ and from \crefpart{build-cor}{blacklist-then-opt} (as $\matching^i(\matching(w))=\matching(\matching(w))$ for any such $w$).

\item
Otherwise, on night $t^i(\tilde{w}^{i+1})$ of $\ifullrun$, by definition of $t^i(\cdot)$,
more than one man serenades under $\tilde{w}^{i+1}$'s window. Denote
by $\tilde{m}$ a man rejected by her on that night. (If more than one such
man exists, choose $\tilde{m}$ arbitrarily among all such men.)
Note that by definition of $t^i(\tilde{w}^{i+1})$,
the man serenading under $\tilde{w}^{i+1}$'s window without being rejected by her on that
night is $\matching^i(\tilde{w}^{i+1})$.
Consider the following timings for $\ifullrun$:
In the first part of the run, denoted by $\firsti$, the algorithm runs normally,
except that the rejection of $\tilde{m}$ by $\tilde{w}^{i+1}$ in favour
of $\matching^i(\tilde{w}^{i+1})$ does not
take place (but on the other hand, neither does she reject $\matching^i(\tilde{w}^{i+1})$).
Thus, by definition of $\tilde{w}^{i+1}$, when $\firsti$
converges, two men ($\tilde{m}$ and $\matching^i(\tilde{w}^{i+1})$)
serenade under $\tilde{w}^{i+1}$'s window,
while each of the remaining $n-2$ men
is the only man serenading under the window of some woman. Thus,
at the end of $\firsti$, there exists a unique woman $\hat{w}$, under
whose window no man serenades. We note that by
definition of $t^i(\tilde{w}^{i+1})$, we have that $\hat{w} \notin T^i$.
In the second part of the run, denoted by
$\secondi$, after $\firsti$
converges, $\tilde{w}^{i+1}$'s rejection of $\tilde{m}$ in favour of
$\matching^i(\tilde{w}^{i+1})$ takes place
and the algorithm runs until it converges once more. Once more,
as the outcome of the Gale-Shapley algorithm is
invariant under timing changes~\citep{Dubins-Freedman}, this timing also yields
$\menopt(\prefs{W}^i,\prefs{M})$.

We denote by $\prefs{W}^{i+1}$ the profile of preference lists for $W$
obtained by applying the
following modifications to $\prefs{W}'$:
\begin{modifications}
\item\label{non-flat-mod-tilde-w}
Along the lines of \cref{build-step},
we set the preference list of $\tilde{w}^{i+1}$ to
be as in $\prefs{W}^i$, but with $\tilde{m}$ promoted to the second
place (immediately following $\matching(\tilde{w}^{i+1})$), maintaining the
internal order of all other men.
\item\label{non-flat-mod-demote}
For each $m \in M$ who satisfies 
$\matching(m)=\matching^{i+1}(m)\ne\matching^i(m)$,
perform the following for each woman $w \notin T^i$ s.t.\ $m$
(strictly) prefers $\matching^i(m)$ over $w$
and $w$ over $\matching(m)$:
\begin{modifications}
\item\label{non-flat-mod-demote-blacklist}
If $w$ has no provisional match at the end of $\firsti$ (i.e.\ if $w=\hat{w}$),
then, once again somewhat along the lines of \cref{build-step},
alter $\hat{w}$'s preferences to blacklist $m$.
\item\label{non-flat-mod-demote-demote}
Otherwise, demote $m$ on $w$'s preference list to somewhere below 
the man with whom $w$ is provisionally matched at the end of $\firsti$.
\end{modifications}
\end{modifications}

We note that \cref{non-flat-mod-demote} is well-defined.
Indeed, since for every such $m$, $\matching^i(m) \ne \matching(m)$,
and thus $\matching^i(m) \in T^i$. Thus, by definition of $\firsti$,
$m$ is provisionally matched with $\matching^i(m)$ at the end
of $\firsti$. Thus, $m$ is not provisionally matched with any $w \notin T^i$
at the end of $\firsti$, and thus \cref{non-flat-mod-demote} does
not specify the demotion of any man ``below himself'' in any preference list.

Consider the following timings, which we soon show to
be well-defined, for $\incfullrun$:
In the first part of the run, denoted by $\firstthird$,
the algorithm runs normally,
except that the rejection of $\matching^i(\tilde{w}^{i+1})$ by $\tilde{w}^{i+1}$ 
in favour of $\tilde{m}$ does not
take place (but on the other hand, neither does she reject $\tilde{m}$).
In the second part of the run, denoted by $\secondthird$, after $\firstthird$ converges,
$\tilde{w}^{i+1}$'s rejection of $\matching^i(\tilde{w}^{i+1})$ in favour of $\tilde{m}$
takes place and the algorithm runs normally once more, except that $\tilde{m}$'s
rejection by $\tilde{w}^{i+1}$ in favour of $\matching(\tilde{w}^{i+1})$
does not take place (but on the either hand, neither does she reject $\matching(\tilde{w}^{i+1})$). In the third and final
part of the run, denoted by $\thirdthird$, after $\secondthird$ converges,
the algorithm runs until it converges.
We now show that this timing
is indeed well-defined, that $\firstthird$ is indistinguishable
from $\firsti$, that $\secondthird$
has the same rejections as $\iprun$, and that $\thirdthird$
has the same rejections as $\secondi$.
(The scenario analysed above, in which on night $t^i(\tilde{w}^{i+1})$ only one man
serenades under $\tilde{w}^{i+1}$'s window, may be viewed in a sense as a special
case of the scenario discussed here, in which $\secondi$ and $\thirdthird$
are both empty.)
From this timing being
well-defined, we yet again have,
as the outcome of the Gale-Shapley algorithm is
invariant under timing changes~\citep{Dubins-Freedman},
that this timing also yields
$\menopt(\prefs{W}^{i+1},\prefs{M})$.

We indeed start by showing that $\firstthird$ is identical
to $\firsti$.
\cref{non-flat-mod-tilde-w} (when compared with the unmodified $\prefs{W}^i$)
has no effect on $\ifullrun$ before $\tilde{m}$
is rejected by $\tilde{w}^{i+1}$ in that run, and thus has no effect on $\firsti$.
As in the proof of \cref{flat-case}, by
\crefpart{build-cor}{diff-then-opt},
the only difference in the preference list of
any other woman $w\ne\tilde{w}^{i+1}$ between
$\prefs{W}^i$ and $\prefs{W}'$ is in the possible promotion of men $m$ who
prefer $\matching^i(m)$ over $w$. As by \cref{non-flat-matching} such men never
serenade under $w$'s window
during $\ifullrun$, such a change has no effect on $\ifullrun$, and thus on
$\firsti$ in particular.
Let $m$ be a man satisfying the conditions of \cref{non-flat-mod-demote};
as \cref{non-flat-mod-demote} only alters the preference, regarding $m$,
of women that $m$ prefers less than $\matching^i(m)$, who, as explained above,
is his provisional match at the end of $\firsti$,
$m$ does not serenade under any of their windows during $\firsti$
and thus this modification has no effect on $\firsti$ either.
By all of the above, as long as $\tilde{w}^{i+1}$ does not reject
$\tilde{m}$ nor $\matching^i(\tilde{w}^{i+1})$, there is no difference
between $\ifullrun$ and
$\incfullrun$. By definition of $\tilde{m}$ and $\matching^i(\tilde{w}^{i+1})$, deferring the
rejection of $\tilde{m}$ by $\tilde{w}^{i+1}$ achieves this effect in $\ifullrun$;
by \cref{non-flat-mod-tilde-w}, deferring the rejection of $\matching^i(\tilde{w}^{i+1})$ by $\tilde{w}^{i+1}$ achieves this effect
in $\incfullrun$. Thus, we have shown
that $\firstthird$ is well-defined and indistinguishable from $\firsti$.

We move on to showing that $\secondthird$ has the same rejections
as $\iprun$. By definition of $\tilde{w}^{i+1}$, each 
$w \in T^i \setminus \{\tilde{w}^{i+1}\}$ is
provisionally matched with $\matching^i(w)$ (with whom she is provisionally
matched on the first night of $\iprun$) at the end of $\firsti$,
and thus also at the end of $\firstthird$. Furthermore,
none of \cref{non-flat-mod-tilde-w,non-flat-mod-demote} apply
to such $w$.
By \crefpart{build-cor}{opt}, $\tilde{w}^{i+1}$
immediately rejects all men approaching her during during $\iprun$, except for
$\matching(\tilde{w}^{i+1})$. By \cref{non-flat-mod-tilde-w}, as
$\tilde{w}^{i+1}$ is provisionally matched with $\tilde{m}$ on the first
night of $\secondthird$, the only man she would not immediately
reject after that night (except for $\tilde{m}$) is $\matching(\tilde{w}^{i+1})$.
Furthermore, since during $\secondthird$ (as during $\iprun$), only one woman is
newly-serenaded-to
every night after the first, we therefore have that the $\secondthird$
would stop immediately if $\matching(\tilde{w}^{i+1})$ 
serenades under $\tilde{w}^{i+1}$'s window (as happens on the last night of
$\iprun$).
Finally, let $m$ be a man rejected during $\iprun$, and let $w \notin T^i$ under
whose window $m$ serenades during $\iprun$. As $\matching^i(w)=\matching(w)$,
and as by definition
$\matching^i(m)\ne\matching(m)$, we have $\matching^i(m)\ne w$, and thus
$m$ serenades under $w$'s window during $\iprun$ following a rejection,
and thus in particular,
by \crefpart{build-cor}{diff-then-opt} and as $w \notin T^i$,
$w$ rejects
$m$ as soon as he serenades under her window during $\iprun$.
As $m$ serenades under $w$'s window following a rejection in $\iprun$,
$m$ prefers $\matching^i(m)$ over $w$;
by \crefpart{build-cor}{diff-then-opt}, $m$ also
prefers $w$ over $\matching(m)$ (as $w\ne\matching(m)$, since $\matching(m) \in T^i$). Furthermore, by \crefpart{build-cor}{reject-then-opt},
$\matching^{i+1}(m)=\matching(m)$. Thus,
\cref{non-flat-mod-demote} was applied
to $w$'s preference list w.r.t.\ $m$. If $w=\hat{w}$, then she thus blacklists
$m$ and would thus immediately reject him if he ever serenades under her window
during $\incfullrun$ (and thus in particular during $\secondthird$);
otherwise, $w\ne\hat{w}$, and as $w$'s provisional
match only improves (according to $\prefs{W}^{i+1}$) during $\incfullrun$,
she would thus immediately reject $m$ if he ever serenades under her window
during $\secondthird$.
We conclude that
any woman who, during $\iprun$, immediately rejects all men serenading under her
window for the first time after the first night, would reject all such men
during $\secondthird$ as well, if any of them serenade
under her window after the first night; we also conclude that any other woman has the same
preference list and initial matching during $\iprun$ and during $\secondthird$,
with the exception of $\tilde{w}^{i+1}$, who would
immediately reject all men except $\matching(\tilde{w}^{i+1})$ in both
$\iprun$ and $\secondthird$, and would accept $\matching(\tilde{w}^{i+1})$, in which case both $\iprun$ and $\secondthird$ stop.
Thus, as the only rejection on the first night of both $\iprun$ and $\secondthird$
is of $\matching^i(\tilde{w}^{i+1})$ by $\tilde{w}^{i+1}$, we conclude that
$\secondthird$ is well-defined and has the same rejections
as $\iprun$.

Finally, we show that $\thirdthird$ has the same rejections
as $\secondi$. Indeed, by definition of $t^i(\tilde{w}^{i+1})$,
all rejections during $\secondi$, in all nights but the first,
are by women in $W \setminus T^i$. As explained above, the provisional matches
of these women are unaltered during $\secondthird$, and thus
are the same on the first nights of $\thirdthird$ and $\secondi$.
As in both these first nights the only rejection is
of $\tilde{m}$ by $\tilde{w}^{i+1}$, it is enough to show that the preferences
of all women in $W \setminus T^i$ agree w.r.t.\ all men provisionally matched to any of $W \setminus T^i$ in both of these first nights,
as well as w.r.t.\ $\tilde{m}$.
By definition of $t^i(\tilde{w}^{i+1})$, we have that all such men $m$ (incl.\
$\tilde{m}$) satisfy $\matching^i(m) \notin T^i$, and thus $\matching^i(m)=\matching(m)$. Thus, \cref{non-flat-mod-demote} does not involve demoting or blacklisting any such $m$. As $\tilde{w}^{i+1} \in T^i$, \cref{non-flat-mod-tilde-w} is irrelevant as well at this point. By \crefpart{build-cor}{diff-then-opt},
the preferences of $W \setminus T^i$ are the same in $\prefs{W}'$ as they
are in $\prefs{W}^i$.
Thus, $\thirdthird$ and $\secondi$ have the same rejections.
We conclude that the rejections occurring in $\incfullrun$ are exactly
those occurring in either of $\ifullrun$ or $\iprun$. Thus, as $\ifullrun$
concludes with the initial matching of $\iprun$,
and as each man is eventually matched with the woman he prefers most out of
those who have not rejected him, we obtain that
$\menopt(\prefs{W}^{i+1},\prefs{M})$ is the matching obtained by first running
$\ifullrun$ and then $\iprun$,
namely $\ipresultmatching=\matching^i$ and \cref{non-flat-matching} holds
for $i+1$.

\cref{non-flat-increase} for $i+1$ follows once more directly from \crefpart{build-cor}{opt-then-opt}.

By \cref{non-flat-idol,non-flat-second-rank} for $i$,
and by \cref{build-cor}(\labelcref{build-cor-opt,build-cor-diff-then-opt}),
in order to show that \cref{non-flat-idol,non-flat-second-rank} hold for $i+1$, it is enough to show that 
for every $w \in W$, \cref{non-flat-mod-tilde-w,non-flat-mod-demote} do not affect the ranking of $\matching(w)$
and of $\matching^{i+1}(w)$ (if $\matching(w)\ne\matching^{i+1}(w)$).
Indeed, \cref{non-flat-mod-tilde-w} affects $\tilde{w}^{i+1}$ only,
but leaves $\matching(\tilde{w}^{i+1})$ at the top of her preference list;
by \crefpart{build-cor}{opt}, this suffices.
\cref{non-flat-mod-demote} only affect the preference list of women
$w$ by demoting men who are not $\matching(w)$ or $\matching^i(w)$ on
$w$'s preference list, and by \crefpart{build-cor}{reject-then-opt},
$\matching^{i+1}(w)$ is one of these two. Thus, \cref{non-flat-idol,non-flat-second-rank} hold
for $i+1$.
\cref{non-flat-compatible} for $i+1$ follows again from \cref{non-flat-compatible}
for $i$ and from \cref{non-flat-idol} for $i+1$, by \cref{still-compatible}
and as $\matching^{i+1}=\ipresultmatching$.

To prove \cref{non-flat-no-double-blacklist} for $i+1$ from
\cref{non-flat-no-double-blacklist} for $i$, we must show once more that every
newly-blacklisted man $m$ is newly-blacklisted by exactly one woman,
is not blacklisted in $\prefs{W}^i$, and
satisfies $\matching^{i+1}(m)=\matching(m)$. Indeed,
by \crefpart{build-cor}{blacklist-then-opt} and by
\cref{non-flat-mod-tilde-w,non-flat-mod-demote}, any man newly-blacklisted
is newly-blacklisted only by $\hat{w}$ and satisfies $\matching^{i+1}(m)=\matching(m) \ne \matching^i(m)$.
Thus, in particular, by \cref{non-flat-no-double-blacklist} for $i$, he
is not blacklisted in $\prefs{W}^i$.

Since by \crefpart{build-cor}{blacklist-then-opt}
and by \cref{non-flat-mod-tilde-w,non-flat-mod-demote}, no
blacklist is changed but perhaps $\hat{w}$'s, 
in order to prove \cref{non-flat-no-blacklist-opt} for $i+1$ given
\cref{non-flat-no-blacklist-opt} for $i$,
it is enough to yet again show that no woman
blacklists $\matching(\hat{w})$ in $\prefs{W}^i$, and that there
exists no woman $w$ with a nonempty blacklist in $\prefs{W}^i$ s.t.\
$\matching(w)$ is blacklisted by $\hat{w}$ in $\prefs{W}^{i+1}$.
By \cref{non-flat-no-blacklist-opt} for $i$, it is enough to prove the former
for the case in which $\hat{w}$ has an empty blacklist
in $\prefs{W}^i$. In this case, by definition of $\hat{w}$, no man serenades
under her window during $\firsti$, and $\secondi$ stops as soon as any man
serenades under her window. Thus, only one man serenades under $\hat{w}$'s
window during $\ifullrun$, and since $\hat{w} \notin T^i$, this man is
$\matching(\hat{w})$. Thus, by \cref{non-flat-single-competitor} for $i$,
$\matching(\hat{w})$ is not blacklisted in $\prefs{W}^i$.
To show the latter, let $w$ be a woman with a nonempty blacklist in $\prefs{W}^i$. By
\cref{non-flat-no-blacklist-opt} for $i$,
$\matching^i(\matching(w))=\matching(\matching(w))$,
and thus, by \cref{non-flat-mod-demote},
$\matching(w)$ is not blacklisted by $\hat{w}$ in $\prefs{W}^{i+1}$
and the proof of \cref{non-flat-no-blacklist-opt} for $i+1$ is complete.

Finally,  by equivalence of $\firstthird$ and $\firsti$, and of
$\thirdthird$ and $\secondi$, any woman $w$ serenaded-to during
$\incfullrun$ solely by $\matching(w)$ is also serenaded-to during $\ifullrun$
solely by him. Thus,
\cref{non-flat-single-competitor} for $i+1$ follows from
\cref{non-flat-single-competitor} for $i$ and from \cref{non-flat-mod-demote}
(as $\matching^i(\matching(w))=\matching(\matching(w))$ for any such $w$).
\end{itemize}

Thus, the proof by induction is complete,
as the process stops by \cref{flat-case-increase} and by finiteness of $W$.
The remainder of the proof follows verbatim as in the proof of \cref{flat-case}.
The $O(n^3)$ time complexity follows by na\"{\i}vely implementing the above
inductive process. We note that the only obstacle to obtaining
$O(n^2)$ time complexity is the need to rerun the Gale-Shapley algorithm
in order to identify $\hat{w}$ on every iteration corresponding to the second
case of the induction step; more formally, the algorithm runs in
$O(n^2 \cdot (1+k))$ time,
where $k$ is the number of iterations of the second case of the induction step
(thus yielding a best-case time complexity of $O(n^2)$, i.e.\ under the conditions
of \cref{flat-case}).
Recall that by \cref{build-base,build-cor}, the number of all iterations is
at most the number of $(\matching^1\rightarrow\matching)$-cycles.
Since $\matching$ is uniformly
distributed given $\prefs{M}$ (and thus given $\matching^1$),
we have that $\matching\circ{\matching^1}^{-1}$ is uniformly distributed in
$S_n$; therefore, the expected number of cycles in $\matching\circ{\matching^1}^{-1}$ is $H_n\eqdef\sum_{j=1}^n \frac{1}{j}$~\citep[p.\ 19]{Arratia-et-al}.
Thus, by \cref{cycles-perm}, we have that $\mathbb{E}[k]\le H_n = O(\log n)$,
and hence the average-case time complexity of the above algorithm is $O(n^2 \log n)$,
as required.

A final implementation note: we observe that as long as $\tilde{w}^{i+1}$ is
chosen as detailed in the proof above, then for any $i$, until time
$t^i(\tilde{w}^{i+1})$, $\ofullrun$ is indistinguishable from $\ifullrun$,
and that after this time, in neither run is any woman in $T^i$
newly-serenaded to. This follows since
$(t^i(\tilde{w}^{i+1}))_{i=1}^d$ is weakly monotone-decreasing
by its definition and by \crefpart{build-cor}{reject-then-opt}.
\end{proof}

\subsection{Proofs of the remaining Theorems from Section~\refintitle{blacklists}}

\begin{proof}[Proof of \cref{inconspicuous-tight}]
We start by showing that \cref{inconspicuous-tight} holds for $n_b=1$
and $l_1=n-1$. Denote the members of $W$ by $w_0,\ldots,w_{n-1}$, and of $M$ ---
by $m_0,\ldots,m_{n-1}$. Let $\matching$ be the matching and $\prefs{M}$ 
be the profile of preference lists for $M$ s.t.\
for every $0 \le j \le n-1$, both $\matching(w_j)\eqdef m_j$
and the preference list of $m_j$ (from most-preferred to
least-preferred) is: $w_{j+1},w_{j+2},\ldots,w_{n-1},w_0,w_1\ldots,w_j$. To show
\cref{inconspicuous-tight-exists}, let $\prefs{W}$ be the profile of preference lists
by which, as in \cref{build-base}, the preference list of $w_0$ is
$m_0$ (with all other men blacklisted), and for every $1\le j\le n-1$,
the preference list
of $w_j$ is first $m_j$, followed immediately by $m_{j-1}$,
followed by all other
men in arbitrary order.
Let $\matching'$ be the matching describing the first night of
$\fullrun$, i.e. $\matching'(m_j)\eqdef w_{(j+1 \mod n)}$ for all $0\le j \le n-1$.
It is straightforward to check that
$\abstractcycle{w_0}=(w_0 \xrightarrow{m_{n-1}} w_{n-1} \xrightarrow{m_{n-2}} w_{n-2} \xrightarrow{m_{n-3}} \cdots w_1 \xrightarrow{m_1} w_0)$. Thus, by \cref{build-base} and by \crefpart{pm-compatible-properties}{opt-if-reject-cycle},
$\resultmatching=\matching$.
(The interested reader may verify that $\prefs{W}$ is the
profile of preference lists obtained by applying the above
proofs/algorithms.)

We now move on to proving \cref{inconspicuous-tight-no-less}. Let $\prefs{W}'$
be a profile of preference lists for $W$ s.t.\
$\menopt(\prefs{W}',\prefs{M})=\matching$.
We must show that there exists a woman whose preference list according to
$\prefs{W}'$ consists of a single man. We consider the timing for $\pfullrun=\pprun$
induced by the recursive implementation~\citep{MW71} of the Gale-Shapley algorithm, i.e.\
as long as any blacklist-induced rejection remains, one such rejection occurs,
its aftermath runs until the algorithm converges once more (we call this
an \emph{iteration}), at which point,
if any other blacklist-induced rejection remain, the process repeats. If, at
any time, a woman is approached by two men she blacklists, she arbitrarily
rejects one of them in favour of the other.

We first note that on each iteration, every woman rejects precisely one man.
Indeed,
assume that an iteration starts with a woman $w_j$ rejecting her provisional
match. By definition of $\prefs{M}$, this provisional match then serenades
under $w_{(j+1 \mod n)}$'s window, who rejects either him or her provisional
match, the rejected man then serenades under $w_{(j+2 \mod n)}$'s window,
who rejects either him or her provisional match, and so forth until
the man rejected by $w_{(j+n-1 \mod n)}$ serenades under the window
of $w_{(j+n \mod n)}=w_j$, who has no provisional match at that point,
and thus the iteration concludes. So, exactly $n$ rejections take place
during each iteration. As $\pfullrun$ yields
$\matching$, each man is rejected by $n-1$ women during this run,
and thus $n\cdot(n-1)$ rejections in total take place during it,
and so the above-described
timing for it consists of $n-1$ iterations.
Let $w$ be the woman whose blacklist-induced rejection of her
then-provisional match $m$ triggers the last iteration of $\pfullrun$. In the previous
iterations, $w$ has already performed $n-2$ rejections (one in each
iteration), and since in the beginning of the last iteration she
is provisionally matched with $m$, whom she blacklists, we conclude that
the $n-2$ men she has already rejected in prior iterations are also blacklisted
by her, for a total of $n-1$ men blacklisted by $w$ and the proof is complete.

The general case of \cref{inconspicuous-tight} follows in a similar way.
Assume w.l.o.g.\ that $l_1+\cdots+l_{n_b}=n-n_b$, by adding $l_i$s
equal to zero (and thus increasing the value of $n_b$) if necessary.
Denote the members of $W$ by $w^i_j$, for $1\le i \le n_b$ and for each such
$i$, for $0 \le j \le l_i$; similarly denote the members of $M$ by
$m^i_j$, for the same values of $i$ and $j$. For every such $i$ and $j$,
we set $\matching(w^i_j)\eqdef m^i_j$ and set the preference list of $m^i_j$
to start with $w^i_{j+1},w^i_{j+2},\ldots,w^i_{l_i},w^i_0,w^i_1,\ldots,w^i_j$,
in this order, followed by all other women in arbitrary order. The proof of
\cref{inconspicuous-tight-exists} is similar to that of the special case,
with $w^i_0$, for every
$1\le i \le n_b$ preferring $m^i_0$ most, and blacklisting 
$m^i_j$ for all $0 \le j \le l_i$; and with $w^i_j$, for every such $i$
and for every $0 < j \le l_i$, preferring $m^i_j$ most, followed immediately
by $m^i_{j-1}$, followed by all other men in arbitrary order. Finally,
\cref{inconspicuous-tight-no-less} follows by performing the analysis
of the special case for each $i$ separately, and by observing that
since $\menopt(\prefs{W}',\prefs{M})=\matching$, no man $m^i_j$ is ever rejected
by $w^i_j$ during $\pfullrun$, and thus
no such man ever serenades under the window of any woman
$w^{i'}_k$, for any $i'\ne i$ and any $k$.
\end{proof}

\begin{proof}[Proof of \cref{inconspicuous-not-all-matched}]
The proof runs along the lines of the proof of \cref{inconspicuous}, with a few differences
that we now survey.

We begin by focusing our attention on $\unmatched{W}$.
By having each women $w \in \unmatched{W}$ blacklist all
men\footnote{It actually suffices for each such $w$ to blacklist all
$m \in \unmatched{M}$ who do not blacklist her, in addition to all 
$m \in \matched{M}$
who prefer $w$ over $\matching(m)$. The proof is left to the
interested reader.}, we guarantee that these women turn out to be unmatched,
and can \emph{de facto} ignore them from this point on.

The base of the inductive argument is as in \cref{inconspicuous}, except
that we place $\unmatched{M}$ last in all preference lists in $\prefs{W}^0$.
(The resulting matching $\matching^1$ is thus the matching that would have been
obtained by considering only $\matched{M}$ in the runs defining $\matching^1$.)
This guarantees that $\unmatched{M}$ are exactly the men unmatched in
$\matching^1$.
Before commencing with the inductive steps of the proof of
\cref{inconspicuous},
we perform the following inductive step, as many times as possible (with
the same induction invariants): if there exists any woman $w \in T^i$
who is serenaded-to by some $\tilde{m} \in \unmatched{M}$
during $\ifullrun$
(equivalently, who is not blacklisted by some such $\tilde{m}$ --- the set of all
such women can be precomputed once at the beginning of the algorithm),
then we denote $\tilde{w}^{i+1} \eqdef w$.
We continue the construction and proof as in the second case of the induction
step of the proof of \cref{inconspicuous} (however w.r.t.\ $\tilde{w}^{i+1}$ and
$\tilde{m}$ as defined here), only noting that all women are
matched at the end of $\firsti$, and therefore \crefpart{non-flat-mod-demote}{blacklist} never takes place. (The interested reader may verify that
\crefpart{non-flat-mod-demote}{demote} is a no-op, and therefore
\cref{non-flat-mod-demote} is redundant in its whole in this case.)
Therefore, no blacklist is changed during this induction step,
and thus when it is no longer possible to conduct any more such steps,
all blacklists in $\prefs{W}^{i+1}$ are empty.
The rest of the results, except for those regarding the worst- and average-case
time complexities,
follow as in the proof of \cref{inconspicuous}.

Let $r$ be the number of women $w$ for whom, when it is no longer possible
to conduct any more induction steps as above, it still holds that
$\matching^{i+1}(w)\ne \matching(w)$.
By \crefpart{build-cor}{opt}, $r \le n_{\matching}-n_h$, yielding the
worst-case time complexity as required, via the same arguments as in
the proof of \cref{inconspicuous}.
For the average-case analysis, note that each induction step as above
causes not only $\tilde{w}^{i+1}$ to be matched with
$\matching(\tilde{w}^{i+1})$,
but actually her entire $(\matching^1\rightarrow\matching)$-cycle to be matched with their partners in $\matching$. As
the expected combined size of
the cycles containing $n_h$ elements in a random permutation uniformly
distributed in $S_{n_{\matching}}$
is $\frac{n_h\cdot(n_{\matching}+1)}{n_h+1}$~\citep{Combined-Length-Distribution},
we have
$\mathbb{E}[r]=n_{\matching}-\frac{n_h\cdot(n_{\matching}+1)}{n_h+1}<\frac{n_{\matching}+1}{n_h+1}$. By Jensen's inequality~\citep{Jensen} and by
concavity of the logarithm function, we have $\mathbb{E}[\log r]\le\log\mathbb{E}[r]=O(\log\frac{n_{\matching}+1}{n_h+1})$. The rest of the analysis
is as in the proof of \cref{inconspicuous}.
\end{proof}

\begin{proof}[Proof of \cref{tightness-not-all-matched}]
The proof runs along the lines of that of \cref{inconspicuous-tight}.
Assume w.l.o.g.\ that $l_1+\cdots+l_{n_b}=n_{\matching}-n_h-n_b$, by adding $l_i$s
equal to zero (and thus increasing the value of $n_b$) if necessary.
Define $H\eqdef \bigl\{w \in \tobematched{W} \mid
\exists m \in \comp{\tobematched{M}} : w \notin B_m\bigr\}$, and note that
by definition, $n_h=|H|$. Let $\{m_w\}_{w \in H}$ be $n_h$ distinct arbitrary
men from $\tobematched{M}$.
For every $m \in M$, define
$P(m)\eqdef\bigl\{w \in \comp{\tobematched{W}} \mid m \in B_w\bigr\}$.
Denote the members of $\tobematched{W} \setminus H$
by $w^i_j$, for $1\le i \le n_b$ and for
each such $i$, for $0 \le j \le l_i$; similarly denote the members of
$\tobematched{M} \setminus \{m_w\}_{w \in H}$ by
$m^i_j$, for the same values of $i$ and $j$. For every such $i$ and $j$,
we set $\matching(w^i_j)\eqdef m^i_j$ and set the preference list of $m^i_j$
to start with $P(m^i_j)$ in arbitrary order,
followed immediately by $w^i_{j+1},w^i_{j+2},\ldots,w^i_{l_i},w^i_0,w^i_1,\ldots,w^i_j$,
in this order, followed by all other women in arbitrary order.
Additionally, for every $w \in H$, we set $\matching(w)=m_w$ and set
the preference list of $m_w$ to start with $P(m_w)$ in arbitrary order,
followed immediately by $w$, followed by all other
women in arbitrary order. Finally, we set the blacklist of each
$m \in \comp{\tobematched{M}}$ to consist of $B_m$, with all other women
appearing in $m$'s preference list in arbitrary order.

The proof of \cref{tightness-not-all-matched-exists} is similar to that of
\crefpart{inconspicuous-tight}{exists},
with $w^i_0$, for every
$1\le i \le n_b$ preferring $m^i_0$ most, and blacklisting 
$m^i_j$ for all $0 < j \le l_i$; and with $w^i_j$, for every such $i$
and for every $0 < j \le l_i$, preferring $m^i_j$ most, followed immediately
by $m^i_{j-1}$, followed by all other men in arbitrary order.
Set the preference list of every $w \in H$ to start with $m_w$, followed
by all other men in arbitrary order. Finally, set the blacklist of each
$w \in \comp{\tobematched{W}}$ to consist of
$B_w \cup \bigl\{m \in \comp{\tobematched{M}}\mid w \notin B_m\}$, with all other
men appearing in $w$'s preference list
in arbitrary order.

For every $w \in H$, $w$'s top choice is $m_w$, and $m_w$'s top choice, out
of all women who do not blacklist him, is $w$ (as all of $P(m_w)$ blacklist him),
and thus $w$ and $m_w$ are matched by $\menopt(\prefs{W},\prefs{M})$
and no man is ever rejected in favour of $m_w$, except by $w$.
We thus \emph{de facto} ignore $H$ and $\{m_w\}_{w \in H}$ henceforth.
Let us consider the timing for $\fullrun$ obtained by deferring
the participation of $\comp{\tobematched{M}}$ until the algorithm converges (we denote
this part of the run by $\firstpart$),  and
only then introducing $\comp{\tobematched{M}}$ into the market, until the algorithm
converges once again (we denote this part of the run by $\secondpart$).
As in the proof of \crefpart{inconspicuous-tight}{exists},
if $\comp{\tobematched{W}}$ were to not participate in $\firstpart$,
then at its end each $m^i_j$ would be matched with $w^i_j$. As
every such $m^i_j$ prefers $w^i_j$ over all of $\comp{\tobematched{W}}$ (except
for $P(m^i_j)$, who all blacklist him), we have that participation
of $\comp{\tobematched{W}}$ in $\firstpart$ would not change the 
resulting matching. Finally, it is straightforward to verify that
at the end of the $\firstpart$, every $w \in \tobematched{W}$
prefers her provisional match over all of $\comp{\tobematched{M}}$, and
every $w \in \comp{\tobematched{W}}$ blacklists all of $\comp{\tobematched{M}}$ who do not
blacklist her. Thus, the provisional matching does not change at any point
during  $\secondpart$, and the
proof of \cref{tightness-not-all-matched-exists} is complete.

To prove \cref{tightness-not-all-matched-no-less},
let $\prefs{W}'$ be a profile of preference lists for $W$ s.t.\
$\menopt(\prefs{W}',\prefs{M})=\matching$.
We commence by examining $\comp{\tobematched{W}}$.
Let $w \in \comp{\tobematched{W}}$ and let $m$ be a man blacklisted by $w$
in $\prefs{W}$.
If $m \in \comp{\tobematched{M}}$, then by the statement of \cref{tightness-not-all-matched-no-less}, we need only consider the case in which $w \notin B_m=B_m(\prefs{M})$;
in this case, the fact that both $w$ and $m$ are
unmatched
by $\menopt(\prefs{W}',\prefs{M})$ implies that $w$ blacklists $m$
in $\prefs{W}'$, as required.
Otherwise, $m \in \tobematched{M}$ and thus, by definition of $\prefs{W}$,
$m \in B_w$; therefore, $w \in P(m)$ and so $m$ prefers $w$ over his match
$\matching(m)$.
As $w$ is unmatched in $\matching$, we thus have that she blacklists $m$
in $\prefs{W}'$ in this case as well.
Finally, the analysis of the blacklists of $\tobematched{W} \setminus H$ is
as in the proof of \crefpart{inconspicuous-tight}{no-less}, since all of $\comp{\tobematched{M}}$
blacklist them, and all of $\{m_w\}_{w\in H}$ prefer their final matches over
them; thus, only $\tobematched{M} \setminus \{m_w\}_{w \in H}$ serenade
during $\pfullrun$
under their windows (and these men are blacklisted by any woman in
$\comp{\tobematched{W}} \cup H$ that they prefer over their
final match).
\end{proof}

\subsection{Proofs of the Theorems from Section~\refintitle{divorces}}

\cref{one-divorce-per-woman} readily follows from the following lemma, phrased
in terms of the vanilla Gale-Shapley algorithm.

\begin{lemma}
Let $\matching'$ and $\matching$ be matchings and let $\prefs{M}$ be a
$(\matching'\rightarrow\matching)$-compatible profile of preference
lists for $M$.
Let $W$ and $M$ be equal-sized sets of women and men, respectively. Define
$n\eqdef|W|=|M|$.
If $\matching' \ne \matching$, then there exists a woman $\tilde{w}$ and a profile $\prefs{W}$ of preference
lists for $W$, s.t.\ 
the only woman with a nonempty blacklist according to $\prefs{W}$ is $\tilde{w}$,
having a blacklist consisting solely of $\matching'(\tilde{w})$, s.t.\
$\resultmatching(\tilde{w})=\matching(\tilde{w})$, and s.t.\
$\prefs{M}$ is $(\resultmatching\rightarrow\matching)$-compatible.
\end{lemma}

\begin{proof}
Denote by $\tildeprefs{W}$ a profile of preference lists for $W$, according
to which each $w \in W$ prefers $\matching(w)$ most, followed immediately
by $\matching'(w)$ (if $\matching'(w)\ne\matching(w)$),
followed by all other men in arbitrary order.
For every $w \in W$ s.t.\ $\matching'(w)\ne\matching(w)$, we denote by
$\tildeprefs{W}^w$ the profile of preference lists for $W$ obtained from
$\tildeprefs{W}$ by having $w$ blacklist $\matching'(w)$. We note that $\tildeprefs{W}^w$
and $\prefs{M}$ are $\matching'$-cycle generating with trigger $w$.
Let $\tilde{w}$ be a woman for whom $\twrejectcycle$ is of
greatest length.

We claim that $\hat{w}\eqdef\matching'(\matching(\tilde{w}))$
rejects some man during $\twrun$. Indeed, assume for contradiction that this
is not the case; we hence show that $\hwrejectcycle$ is of greater length
than $\twrejectcycle$. Indeed, we claim that $\hwrejectcycle$
has prefix
$(\hat{w} \xrightarrow{\matching(\tilde{w})} \tilde{w}
\xrightarrow{m_1} w_2 \xrightarrow{m_2} w_3 \xrightarrow{m_3} \cdots w_{d-1} \xrightarrow{m_{d-1}}$, where
$\generalcycle\eqdef\twrejectcycle$.
By definition of $\tildeprefs{W}$, and since by definition $\matching'(\hat{w})=
\matching(\tilde{w})$, we directly have that $\hwrejectcycle$
has prefix $(\hat{w}\xrightarrow{\matching(\tilde{w})}\tilde{w}\xrightarrow{m_1}$.
As $\tilde{w}$ is not serenaded-to
except on the first and last nights of $\twrun$, and as,
by assumption, $\hat{w}$ is serenaded-to
solely by $\matching'(\hat{w})$ throughout $\twrun$, the following nights of
$\hwrun$ have the exact same rejections as the second to one-before-last nights
of $\twrun$, respectively,
and thus indeed $\hwrejectcycle$ has the above prefix --- a contradiction.

If $\hat{w}$ rejects $\matching'(\hat{w})\!=\!\matching(\tilde{w})$ during
$\twrun$, then by definition of $\tildeprefs{W}^{\tilde{w}}$ we have
$\wresultmatching(\matching(\tilde{w}))=\tilde{w}$, and the proof is complete by setting
$\prefs{W}\eqdef\tildeprefs{W}^{\tilde{w}}$. Otherwise,
$\hat{w}$ rejects some man $\tilde{m}$ in favour of $\matching'(\hat{w})$
during $\twrun$. Let $\prefs{W}$ be the profile of preference
lists obtained from $\tildeprefs{W}^{\tilde{w}}$ by promoting $\tilde{m}$ to be second on 
$\hat{w}$'s preference list (immediately following $\matching(\hat{w})$
and immediately followed by $\matching'(\hat{w})$). Thus, in $\run$,
$\matching'(\hat{w})=\matching(\tilde{w})$ is rejected by $\hat{w}$ in favour of
$\tilde{m}$ and once again, by definition of $\prefs{W}$, we have
$\resultmatching(\matching(\tilde{w}))=\tilde{w}$ and the proof is complete.
(Either way, as no woman $w$ rejects $\matching(w)$ during $\run$, we have that
$\prefs{M}$ is $(\resultmatching\rightarrow\matching)$-compatible, as required.)
\end{proof}

\begin{proof}[Proof of \cref{divorce-tight}]
The proof runs parallel to that of the special case of \crefpart{inconspicuous-tight}{no-less}, by noting that the first season concludes as soon as it
commences, and that each consecutive season precisely corresponds to one
iteration, as
defined there, and since, as explained there, $n-1$ iterations are
required.
\end{proof}

\end{document}